\documentclass[11pt]{article}

\usepackage[margin=1in]{geometry}
\usepackage{amsmath,amssymb,amsfonts,amsthm}
\usepackage{graphicx}
\usepackage{booktabs}
\usepackage{natbib}
\usepackage[hyphens]{url}
\usepackage{hyperref}
\usepackage{subcaption}
\usepackage{algorithm}
\usepackage{algorithmic}

\newtheorem{proposition}{Proposition}
\newtheorem{theorem}{Theorem}
\newtheorem{observation}{Observation}

\title{Correlated Chance Sampling for Monte Carlo Counterfactual Regret Minimization}
\author{Boning Li\thanks{IIIS, Tsinghua University. Email: li-bn22@mails.tsinghua.edu.cn}
\and Yu Chen\thanks{IIIS, Tsinghua University.}
\and Longbo Huang\thanks{IIIS, Tsinghua University. Email: longbohuang@tsinghua.edu.cn. Corresponding author.}}
\date{\today}

\begin{document}

\maketitle

\begin{abstract}
Monte Carlo Counterfactual Regret Minimization (MCCFR) repeatedly allocates chance outcomes while its strategy evolves, yet standard sampling draws those outcomes independently on every visit. We introduce Correlated Chance Sampling MCCFR (CCS-MCCFR), a drop-in replacement that assigns each concrete chance node a persistent randomized Weyl stream and maps its phases through the node's chance distribution. Each fixed-index draw has the correct marginal law, while the first $N$ draws consumed during $N$ visits to one concrete node achieve deterministic local frequency error $O(\!\log(N+1)/N)$, compared with the $O(N^{-1/2})$ expected scale of i.i.d. frequencies. We further establish unbiasedness along fixed strategy trajectories, isolate adaptive phase selection through a conditional scalar bound, and show that a per-traversal reset variant retains the standard $O(1/\sqrt{T})$ External Sampling guarantee. In paired experiments, CCS-MCCFR reduces final exploitability by 19.05\% to 34.01\% across Kuhn poker and four Leduc poker configurations, with every paired-bootstrap confidence interval above zero, and by a significant 4.27\% on Goofspiel-4. The gain survives to 3M Leduc node touches and combines with Linear CFR to reach the lowest measured exploitability. The sampler introduces no new hyperparameters and no measurable time overhead, so CCS-MCCFR turns a one-line change to the chance sampler into explicit local guarantees and large exploitability reductions across tabular poker.
\end{abstract}

\section{Introduction}

Imperfect-information games such as poker are central benchmarks for sequential decision making under uncertainty. Counterfactual Regret Minimization (CFR) \cite{zinkevich2007regret} is a foundational algorithm for computing approximate Nash equilibria in these games, and CFR-based equilibrium computation and real-time solving contributed to expert or superhuman performance in heads-up no-limit poker through DeepStack and Libratus, and later to superhuman multiplayer poker through Pluribus \citep{moravcik2017deepstack,brown2018superhuman,brown2019superhuman}; the same regret-minimization core also underpins unified perfect- and imperfect-information learning \citep{schmid2023student} and remains the reference against which recent language-model poker agents are measured \citep{li2026pokerskill}. Monte Carlo Counterfactual Regret Minimization (MCCFR) \cite{lanctot2009monte} replaces full-tree updates with sampled traversals. The quality of those sampled updates depends strongly on sampling variance, including variation from card deals and other chance events.

Existing variance-reduction techniques for MCCFR primarily modify sampled estimators. VR-MCCFR \cite{schmid2019variance} introduces control-variate baselines for sampled counterfactual values. Beyond these MCCFR-specific methods, ESCHER \cite{mcaleer2023escher} is a model-free deep CFR-style algorithm that replaces rollout-based importance-weighted regret targets with a learned history-value estimator and a fixed update-player sampling policy. Neither line directly balances how outcomes at the same concrete chance node are allocated across repeated visits. We study this temporal structure and its interaction with adaptive MCCFR updates.

The basic idea is simple. Independent sampling can repeatedly select some outcomes and temporarily miss others when MCCFR revisits the same concrete chance node. A persistent randomized Weyl stream spreads those visits more evenly across the node's chance distribution \citep{niederreiter1992random}. When downstream values are fixed or sufficiently stable, this local balance can reduce accumulated chance error; under adaptive regret updates, the effect additionally depends on how sampled outcomes couple to evolving downstream values. This motivates the revisit, symmetry, and private-information diagnostics used below.

We propose \textbf{Correlated Chance Sampling MCCFR} (CCS-MCCFR), a drop-in modification that places one persistent randomized low-discrepancy stream at each concrete chance node \citep{niederreiter1992random}. Random shifting preserves fixed-index marginal correctness while retaining low discrepancy \citep{cranley1976randomization,owen1997monte}. The sampler leaves the chance law, External-Sampling estimator, and regret update formula unchanged, adds no hyperparameters, and costs no measurable solver time (Appendix~\ref{app:extra}). It is the cheapest intervention we are aware of that yields exploitability reductions of this size in tabular poker.

The strongest results occur in tabular poker. CCS-MCCFR lowers final exploitability by 27.64\% on Kuhn, 24.59\% on standard Leduc, and 27.65\%, 34.01\%, and 19.05\% in a controlled 6/10/12-card Leduc deck expansion, with all corresponding paired-bootstrap confidence intervals above zero. On Leduc the 24.5\% reduction persists at 3M node touches, and Linear CFR plus CCS-MCCFR is the lowest cell in a six-configuration composition grid, 42.59\% below vanilla updates with i.i.d.\ chance. Goofspiel-4 yields a smaller but significant 4.27\% reduction. Liar's Dice, reduced Flop Hold'em, and four Libratus Turn and River endgames place the endpoint near vanilla, with every paired interval crossing zero. Revisit exposure, symmetry, and private-information coupling organize this pattern as empirical moderators and match the theoretical role of evolving downstream values; Appendix~\ref{app:scope} states the scope of these claims.

\textbf{Contributions:}
\begin{itemize}
\item We introduce CCS-MCCFR, the first persistent randomized low-discrepancy chance sampler for MCCFR regret estimation, and demonstrate significant exploitability reductions on Kuhn, standard Leduc, and a controlled Leduc deck family, including reductions from 19.05\% to 34.01\% across the three deck configurations. The sampler is a drop-in one-line change with no new hyperparameters.
\item We establish fixed-index marginal correctness, unbiasedness along fixed strategy trajectories, and deterministic $O(\!\log(N+1)/N)$ local frequency discrepancy for the first $N$ draws consumed by one concrete node. We also bound the adaptive conditional scalar phase-selection term and prove that per-traversal resetting retains the standard $O(1/\sqrt{T})$ External Sampling guarantee.
\item We map the sampler's empirical operating regime through a persistence ablation, revisit and symmetry diagnostics, controlled private-information configurations, update-rule composition, and boundary tests spanning Goofspiel, Liar's Dice, reduced Flop Hold'em, and HUNL endgames.
\end{itemize}

\section{Related Work}
\label{sec:related}

\textbf{MCCFR and sampling schemes.} Monte Carlo CFR replaces exact counterfactual regret increments with sampled estimates while preserving the no-regret structure of CFR \citep{lanctot2009monte}. Outcome Sampling samples a single trajectory and uses importance corrections. External Sampling samples opponent actions and chance outcomes while enumerating the traversing player's actions, and is the external sampling regime used throughout this paper. Generalized sampling and probing analyze how different sampling blocks affect variance \citep{gibson2012generalized}, average strategy sampling changes which information is collected for the average strategy \citep{burch2012efficient}, and public chance sampling shares sampled public chance events across compatible information sets \citep{johanson2012efficient}. More recent work continues along this axis: information-set sampling improves the training efficiency of neural imperfect-information agents \citep{bertram2024efficiently}, and counterfactual-value-based fictitious play accelerates equilibrium convergence in Monte Carlo settings \citep{qi2024accelerating}. These methods choose what is sampled and how the sampled estimator is weighted. CCS-MCCFR keeps the chance distribution and the external sampling estimator fixed, but changes the temporal structure of the chance draws by attaching a persistent low discrepancy stream to each concrete chance node across iterations.

\textbf{Variance reduction and regret update accelerators.} Control variates, baselines, and common random numbers form the standard toolkit for reducing the variance of Monte Carlo estimates \citep{mohamed2020monte}. VR-MCCFR applies control-variate state-action baselines, learned online in its main instantiation, to sampled counterfactual values \citep{schmid2019variance}, and low- and zero-variance baselines identify baseline choices for extensive-form games \citep{davis2020low}. Variance also receives attention on the optimization side, where stochastic gradient estimates for equilibrium approximation are stabilized directly \citep{meng2025reducing}. Baselines also appear in policy-gradient form for partially observable multiagent settings \citep{srinivasan2018actor}, and regret targets can be estimated by function approximation \citep{waugh2015solving}. A related deep-CFR line \citep{brown2019deep,steinberger2020dream,brown2020combining} changes the regret target more substantially: ESCHER is a model-free deep CFR-style algorithm that combines a learned history-value estimator with a fixed update-player sampling policy, removing the OS-MCCFR-derived reach and continuation importance corrections from its own estimator \citep{mcaleer2023escher}. Our intervention instead leaves the External-Sampling estimator algebra unchanged and alters the temporal allocation of repeated chance outcomes at each concrete node. Thus VR-MCCFR and CCS-MCCFR act on distinct algorithmic axes: the former corrects sampled value estimates, whereas the latter changes the cross-visit schedule of chance outcomes while preserving the chance law. Because they act on different axes, they can be applied together, and we measure that composition rather than assume it. Another line changes regret updates or averaging, including CFR+ \citep{tammelin2014solving} and Discounted CFR with its Linear CFR special case \citep{brown2019solving}. This line remains active: predictive Blackwell approachability connects regret matching to mirror descent \citep{farina2021faster}, discount factors are made dynamic \citep{xu2024dynamic} or scheduled over the run \citep{zhang2026faster}, step sizes are made asymmetric across players \citep{meng2026faster}, regret matching is made parameter-free \citep{meng2026parameterfree}, and the discounted updates carry over to deep CFR \citep{xu2026deep}. These are modular implementation axes, but an update rule changes the downstream values paired with the chance stream. Section~\ref{sec:orthogonality} therefore measures their empirical composition directly, and Appendix~\ref{app:es-control-variate} reports a restricted External-Sampling control-variate composition test motivated by VR-MCCFR.

\textbf{Low-discrepancy sequences and randomized QMC.} Low-discrepancy sequences are classical tools for reducing integration error \citep{niederreiter1992random}, with Sobol and Halton sequences as standard examples \citep{sobol1967distribution,halton1960efficiency}. Random shifts and scrambling are standard randomized QMC mechanisms for preserving marginal correctness while exploiting low discrepancy \citep{cranley1976randomization,owen1997monte}. Systematic sampling similarly balances coverage in classical survey and simulation settings \citep{madow1944theory}. In machine learning these constructions replace i.i.d.\ draws inside kernel feature maps \citep{avron2016quasi} and inside the stochastic gradient estimates of variational inference \citep{buchholz2018quasi}, the latter being the closest precedent for embedding a low-discrepancy stream in an adaptive iterative procedure. CCS-MCCFR uses the simplest one dimensional randomized construction, a shifted Weyl rotation, and binds one stream to each concrete chance node. The primitive is classical; the placement inside adaptive repeated MCCFR traversal is the point that creates both the benefit and the phase selection issue analyzed in Section~\ref{sec:method}.

\textbf{QMC and correlated sampling in games and learning.} \citet{lisy2015online} applied QMC ideas to online Monte Carlo search in imperfect-information games, making it the closest QMC precedent in this area, and that search setting has since been developed through continual resolving \citep{sustr2019monte}, its soundness analysis \citep{sustr2020sound}, and the formal models that underpin both \citep{kovarik2022rethinking}. Their target is online sampling at search time, whereas our target is repeated tabular MCCFR regret estimation with persistent per-chance-node streams. Common random numbers in reinforcement learning \citep{peshkin2002learning} and antithetic sampling in random search \citep{mania2018simple} also use correlation to reduce estimator noise, but they do not address chance sampling inside extensive form game regret minimization. We contribute a new placement of randomized low discrepancy streams within MCCFR, an explicit characterization of adaptive phase selection, and paired experiments that identify both favorable configurations and empirical regime boundaries.

\textbf{Complementary axes of speedup.} Faster game solving is also pursued by shrinking the tree that must be traversed and by traversing it faster. Action abstraction can be selected by reinforcement learning \citep{li2024rl}, online pruning and abstraction can be applied during solving \citep{li2025efficient}, and information abstraction can be constructed to transfer across games \citep{li2026effective}. On the implementation side, CFR is parallelized for real-time use \citep{li2026parallel} and mapped onto accelerators for specific games \citep{baghal2025pasur}, while meta-learning tunes the self-play regret minimizer itself \citep{sychrovsky2025meta}. Each of these axes decides which nodes are visited or how quickly a visit executes, whereas CCS-MCCFR changes the outcome assigned to a chance node across visits, so it composes with them rather than competing.

\section{Background}

\paragraph{Extensive Form Games and CFR.}

An \textit{extensive-form game} (EFG) is defined by a game tree with decision nodes controlled by players, chance nodes governed by fixed probability distributions, and terminal nodes with payoffs. At decision nodes, player $i$ observes an \textit{information set} $I \in \mathcal{I}_i$ and selects an action $a \in A(I)$. A behavioral strategy $\sigma_i$ assigns a probability distribution over actions at each information set.

Let $\pi^\sigma(h)$ be the reach probability of history $h$ under strategy profile $\sigma$. We write $\pi_i^\sigma(h)$ for player $i$'s contribution to this reach probability and $\pi_{-i}^\sigma(h)$ for the product of chance and all players other than $i$. With this convention, the counterfactual value of action $a$ at information set $I$ is
\[
v_i^\sigma(I,a)=\sum_{h\in I}\pi_{-i}^\sigma(h)\sum_{z\sqsupset ha}\pi^\sigma(ha,z)u_i(z),
\]
where $u_i(z)$ is player $i$'s terminal utility and $\pi^\sigma(ha,z)$ is the continuation probability from history $ha$ to terminal history $z$. The value of the current strategy at $I$ is $v_i^\sigma(I,\sigma_i)=\sum_{a\in A(I)}\sigma_i(I,a)v_i^\sigma(I,a)$.

The instantaneous regret for action $a$ at iteration $t$ is
\[
r^t(I,a)=v_i^{\sigma^t}(I,a)-v_i^{\sigma^t}(I,\sigma_i^t).
\]
CFR accumulates $R^t(I,a)=\sum_{\tau=1}^t r^\tau(I,a)$ and updates strategies by regret matching, $\sigma_i^{t+1}(I,a)\propto \max(R^t(I,a),0)$. The output is the reach weighted average strategy, defined at each information set by
\[
\bar\sigma_i^T(I,a)=
\frac{\sum_{t=1}^T \pi_i^{\sigma^t}(I)\sigma_i^t(I,a)}
{\sum_{t=1}^T \pi_i^{\sigma^t}(I)},
\]
with an arbitrary convention when the denominator is zero. In two player zero sum games, vanishing average external regret implies convergence of this average strategy to a Nash equilibrium at the standard $O(T^{-1/2})$ regret rate \citep{zinkevich2007regret}.

\paragraph{Monte Carlo CFR}
\label{sec:mccfr-background}

Computing exact counterfactual values requires traversing the full game tree. MCCFR replaces exact regret increments with sampled estimates $\tilde r^t(I,a)$ that are unbiased under the chosen sampling scheme \citep{lanctot2009monte}. We use External Sampling. For the traversing player, all legal actions at each reached information set are evaluated; for opponents and chance, one branch is sampled according to the opponent strategy or the fixed chance distribution. In vanilla External Sampling, each visit to a concrete chance node draws independently from the game's chance law $f_c$, such as the uniform distribution over legal card deals. CCS-MCCFR changes only this final draw rule and leaves the external sampling regret estimator unchanged.

\section{The CCS-MCCFR Sampler}
\label{sec:method}

\paragraph{Motivation}

Consider a chance node $c$ with outcome distribution $f_c$ (e.g., dealing a flop card uniformly from 48 remaining cards). In vanilla MCCFR, each visit to $c$ samples an outcome $o \sim f_c$ independently. Let $N$ denote the number of draws consumed by this one concrete node's stream, equivalently its first $N$ visits; $N$ is distinct from the global iteration horizon $T$ and from the experimental node-touch budget. Over these $N$ visits, the empirical frequency $\hat{p}_k^{(N)}$ of outcome $k$ has expected error $\mathbb{E}|\hat{p}_k^{(N)} - p_k| = \Theta(N^{-1/2})$ under i.i.d. sampling. When downstream per-outcome values are fixed or effectively oblivious to the sampled phase, this frequency error propagates linearly into the corresponding chance-node value estimates. The fully adaptive case, where those values themselves evolve with the sampled outcomes, is analyzed separately below.

\textbf{Local frequency-balancing idea:} Low discrepancy phases spread the unweighted outcomes of a repeatedly visited node more evenly over time than i.i.d. draws. The experiments test when this improved temporal allocation combines with evolving downstream values to reduce exploitability.

\paragraph{CCS-MCCFR via Persistent Weyl Sequences}

For each concrete chance node $c$, we associate a persistent Weyl sequence. Let $N_c$ be the number of previous visits to $c$. On the next visit, the sampler uses
\[
u_{c,N_c}=(\phi_c+N_c g)\bmod 1,
\]
where $g=(\sqrt{5}-1)/2$ and $\phi_c\sim\mathrm{Uniform}[0,1)$ is drawn once at initialization. It then maps $u_{c,N_c}$ to an outcome through the quantile function $\Phi_c^{-1}$ of the chance distribution $f_c$, returns $o=\Phi_c^{-1}(u_{c,N_c})$, and increments $N_c$ by one. For a discrete distribution $f_c=(p_1,\ldots,p_m)$, the quantile rule selects outcome $k$ when $\sum_{j<k}p_j\le u_{c,N_c}<\sum_{j\le k}p_j$.

\textbf{Why the golden ratio?} The golden ratio has especially poor rational approximations, with continued fraction $[0;1,1,1,\ldots]$, giving favorable one-dimensional discrepancy constants for Kronecker or Weyl rotations. This yields small deviations from uniform coverage over intervals and leads to the frequency error bounds below \citep{niederreiter1992random}.

\paragraph{Marginal Correctness and Bias Analysis}

\begin{proposition}[Marginal correctness]
\label{prop:unbiased}
For any fixed $n\ge 0$, the marginal distribution of $o_{c,n}=\Phi_c^{-1}(u_{c,n})$ is $f_c$.
\end{proposition}

\begin{proof}
Since $\phi_c \sim \mathrm{Uniform}[0,1)$ and addition by a constant modulo one is measure-preserving on the circle, $u_{c,n}=(\phi_c+ng)\bmod 1$ is uniform for any fixed $n\ge0$. The quantile map $\Phi_c^{-1}(u_{c,n})$ then has distribution $f_c$: for the discrete threshold rule above, $\Pr[o_{c,n}=k]=\Pr[\sum_{j<k}p_j\le u_{c,n}<\sum_{j\le k}p_j]=p_k$. The same conclusion holds for any random index $N$ that is independent of $\phi_c$, by conditioning on $N$. The converse fails, since a phase dependent index can render $u_{c,N}$ non uniform, which is exactly the adaptive bias analyzed below.
\end{proof}

\textbf{Adaptive tree traversal.} Proposition~\ref{prop:unbiased} establishes marginal correctness for a fixed visit index. We next separate the fixed-trajectory case from the fully adaptive case. A concrete chance node $c$ is a fixed tree location with its own private phase stream $\phi_c$, and $N_c(t)$ denotes the number of previous visits to that exact node before the current draw at iteration $t$.

\begin{proposition}[Trajectory unbiasedness]
\label{prop:traj}
Along any fixed, data independent strategy sequence $\sigma^{1:T}$, the CCS-MCCFR regret estimates satisfy $\mathbb{E}[\tilde{r}^t(I,a) \mid \sigma^{1:t}] = r^t_{\sigma^t}(I,a)$ for all $t$.
\end{proposition}

\begin{proof}
Fix the strategy sequence $\sigma^{1:T}$. In an acyclic extensive-form tree, a concrete chance node is reached at most once during a single traversal. Whether traversal reaches $c$ at iteration $t$ is determined by the fixed strategies and by sampled actions and chance outcomes at ancestor or other concrete nodes. It is not determined by the outcome sampled at $c$, because that outcome is observed only after $c$ has been reached. Under the fixed strategy sequence, earlier outcomes sampled at $c$ also do not change later strategies or later reach decisions. Hence the prior visit count $N_c(t)$ is independent of the phase $\phi_c$ that determines the current marginal draw. Conditional on $N_c(t)=n$, Proposition~\ref{prop:unbiased} gives $u_{c,n}\sim\mathrm{Uniform}[0,1)$ and $\Phi_c^{-1}(u_{c,n})\sim f_c$. Distinct concrete chance nodes have independent phase shifts, so the within traversal joint chance law is the same product law as in independent External Sampling. The standard unbiasedness of the External Sampling MCCFR estimator therefore carries over, giving $\mathbb{E}[\tilde{r}^t(I,a)\mid \sigma^{1:t}]=r^t_{\sigma^t}(I,a)$.
\end{proof}

Adaptive MCCFR introduces an endogenous feedback loop. The realized strategy $\sigma^t$ depends on regrets accumulated from past outcomes drawn at $c$, so $N_c(t)$ can covary with the persistent phase $\phi_c$ through the strategy history. The next proposition isolates this phase-selection channel for a bounded scalar contribution.

\textbf{A conditional scalar phase-selection term.} Fix a traversal at iteration $t$ and let $R_{c,t}$ denote the event that this traversal reaches, and hence consumes, node $c$; on this event the draw at $c$ consumes phase $u_{c,N_c(t)}=(\phi_c+N_c(t)g)\bmod 1$. All quantities below are defined on $R_{c,t}$ and all expectations are conditional on it; note that $R_{c,t}$ itself may covary with $\phi_c$ through the strategy history, and this reach selection is part of what the conditional law of the consumed phase absorbs. Let $O_c$ be the outcome set at $c$. After further conditioning on all randomness outside the selected phase at $c$, write $g_t:O_c\to[-G,G]$ for the bounded downstream scalar contribution to the regret or value term under consideration when the current outcome is set to $o\in O_c$. In adaptive MCCFR, $g_t$ is a random function and may depend on $\phi_c$ through previous outcomes and regrets. Writing $\bar g_t=\mathbb{E}_{o\sim f_c}[g_t(o)]$ for the true chance mean of this random function, define
\[
B_c(t):=\mathbb{E}\!\left[g_t\!\left(\Phi_c^{-1}(u_{c,N_c(t)})\right)\,\middle|\,R_{c,t}\right]-\mathbb{E}[\bar g_t\mid R_{c,t}],
\]
where the expectations also average over the phase at $c$ and independent sampling elsewhere. The bias can arise because the selected phase may deviate from uniform after conditioning on reaching $c$ and on $g_t$, and because $g_t$ itself can covary with the selected phase. It vanishes along a fixed strategy trajectory, where Proposition~\ref{prop:traj} applies.

\begin{proposition}[Conditional scalar phase-selection bound]
\label{prop:biasbound}
Assume the downstream scalar contribution satisfies $\|g_t\|_\infty\le G$, and let
\[
\delta_{c,t}:=\mathbb{E}_{g_t}\!\Big[\,\mathrm{TV}\!\big(\mathrm{Law}(u_{c,N_c(t)}\mid g_t,R_{c,t}),\,\mathrm{Unif}[0,1)\big)\,\Big|\,R_{c,t}\Big]
\]
be the expected \emph{conditional} total variation between the law of the selected phase given $g_t$ (on the reach event $R_{c,t}$) and the uniform law. Then
\[
|B_c(t)|\le 2G\,\delta_{c,t} .
\]
\end{proposition}
Here $G$ is a uniform bound on the scalar term being analyzed. For a value term under payoffs in an interval of width $2\Delta$, one may take $G=\Delta$ after centering; for a regret difference term, one may take the corresponding regret increment bound. Total variation is normalized as $\mathrm{TV}(P,Q):=\sup_{A}|P(A)-Q(A)|$, under which $|\mathbb{E}_P[\psi]-\mathbb{E}_Q[\psi]|\le 2\|\psi\|_\infty\,\mathrm{TV}(P,Q)$ for bounded $\psi$.

\begin{proof}
Work on the event $R_{c,t}$; all conditioning below includes it. Condition further on $g_t$ (a bounded, $\sigma^t$-measurable function). Given $g_t=\hat g$, set $\psi:=\hat g\circ\Phi_c^{-1}$, so $\|\psi\|_\infty\le G$; the true mean $\bar g_t=\sum_o f_c(o)\hat g(o)$ is a constant equal to $\mathbb{E}_{U\sim\mathrm{Unif}}[\psi(U)]$ by the quantile pushforward of Proposition~\ref{prop:unbiased} (an \emph{independent} uniform phase maps to law $f_c$). With $P_{\hat g}:=\mathrm{Law}(u_{c,N_c(t)}\mid g_t=\hat g,\,R_{c,t})$,
\[
\big|\,\mathbb{E}[g_t(\Phi_c^{-1}(u_{c,N_c(t)}))-\bar g_t\mid g_t=\hat g,\,R_{c,t}]\,\big|
=\big|\mathbb{E}_{P_{\hat g}}[\psi]-\mathbb{E}_{\mathrm{Unif}}[\psi]\big|
\le 2\|\psi\|_\infty\,\mathrm{TV}(P_{\hat g},\mathrm{Unif}).
\]
Taking expectation over $g_t$ given $R_{c,t}$ and applying the triangle inequality (tower property) gives $|B_c(t)|\le 2G\,\delta_{c,t}$. Conditioning on $g_t$ and on $R_{c,t}$ absorbs both the coupling between $g_t$ and the phase and the reach selection, so no independence between $g_t$ and $\phi_c$ is assumed.
\end{proof}

\textbf{What the bound controls.} Proposition~\ref{prop:biasbound} pins the entire adaptive effect at one node and time to a single conditional joint-process quantity, $\delta_{c,t}$, which measures how far the consumed phase departs from uniform once the downstream contribution and the reach event are known. Proposition~\ref{prop:traj} sets this quantity to zero along fixed trajectories, and Theorem~\ref{thm:reset} gives a per-traversal reset construction that retains the standard External Sampling guarantee. Appendix~\ref{app:adaptive-frequency} measures the adaptive frequency behavior and large-budget endpoints directly, and Appendix~\ref{app:scope} states the reach of the inequality.

The randomization $\phi_c$ is essential: without it, the sequence would be deterministic and potentially biased for particular game structures. Randomizing an otherwise deterministic low-discrepancy construction is also the standard route by which randomized QMC recovers Monte Carlo style marginal correctness and error estimates in integration settings \citep{cranley1976randomization,owen1997monte}; in adaptive MCCFR, the persistent reuse of the same phase creates the additional phase selection issue analyzed above.

\paragraph{Local Discrepancy and a Temporal-Covariance Hypothesis}

The key advantage of Weyl sequences lies in their \textit{low discrepancy}. Recall the star discrepancy of a finite point set $u_0,\ldots,u_{N-1}\in[0,1)$,
$D_N^* := \sup_{b\in(0,1]} \big|\tfrac{1}{N}\#\{n<N: u_n\in[0,b)\} - b\big|$.
Kronecker sequences $\{ng \bmod 1\}$ with a badly approximable rotation number satisfy $N D_N^* = O(\log N)$, and the golden ratio, whose continued fraction is $[0;1,1,1,\ldots]$, attains the most favorable constant \citep{niederreiter1992random}. This yields the following deterministic frequency guarantee for our per node streams.

\begin{theorem}[Deterministic per-node frequency error]
\label{thm:freq}
Fix a concrete chance node $c$ with outcome distribution $f_c=(p_1,\ldots,p_m)$, and let $\hat p_k^{(N)}$ be the empirical frequency of outcome $k$ among the first $N$ draws of the stream $u_{c,n}=(\phi_c+ng)\bmod 1$, $n=0,\ldots,N-1$, mapped through the quantile rule of Section~\ref{sec:method}. Then there is an absolute constant $C$, depending only on $g=(\sqrt5-1)/2$, such that for every initial phase $\phi_c\in[0,1)$ and every $N\ge 1$,
\[
\max_{1\le k\le m}\big|\hat p_k^{(N)}-p_k\big| \;\le\; \frac{C\log(N+1)}{N},
\qquad
\sum_{k=1}^m\big|\hat p_k^{(N)}-p_k\big| \;\le\; \frac{C\,m\log(N+1)}{N},
\]
whereas i.i.d.\ sampling incurs expected error $\Theta(N^{-1/2})$ per outcome with $p_k\in(0,1)$.
\end{theorem}

\begin{proof}
Outcome $k$ is selected when $u_{c,n}$ falls in the half open interval
\[
I_k=\left[\sum_{j<k}p_j,\sum_{j\le k}p_j\right).
\]
Therefore, $|\hat p_k^{(N)}-p_k|$ is bounded by the extreme (interval) discrepancy of the point set and hence by $2D_N^*$. The point set is the rotation of $\{ng\bmod 1\}_{n<N}$ by $\phi_c$; discrepancy over wrapped arcs is invariant under rotation, and the counting error of any interval is bounded by that of at most two wrapped arcs, so the shift costs at most a constant factor uniformly in $\phi_c$. Finally, for $g$ with bounded partial quotients the Kronecker sequence satisfies $N D_N^* = O(\log(N+1))$ with an absolute constant \citep{niederreiter1992random}. Summing over the $m$ outcomes gives the second bound.
\end{proof}

\textbf{What the theorem controls.} The guarantee is per-node and per-stream: it bounds the unweighted outcome counts in the first $N$ draws consumed by one concrete node, where $N$ is that node's visit count rather than the global iteration count $T$. Because the bound holds uniformly over phases, pauses between visits leave the consumed prefix and its discrepancy unchanged, which is what makes a stream persistent across iterations well behaved.

For fixed per-outcome values, this count discrepancy immediately controls the corresponding unweighted sample mean. MCCFR instead supplies a time-varying scalar contribution $g_t(o)$, so the cumulative chance contribution depends on how the outcome indicators covary across time with these evolving values. I.i.d. draws remove temporal covariance between distinct draws, whereas a persistent Weyl stream deliberately introduces temporal dependence.

Our mechanism hypothesis is that this dependence produces negative cumulative covariance, or outright cancellation, once the outcomes are projected through sufficiently stable downstream values. Three structural features follow from it and drive the experimental design: repeated visits supply the longer streams on which cancellation can accumulate, symmetry can cancel chance terms before temporal balancing acts on them, and low revisit counts leave little temporal structure to exploit. Section~\ref{sec:experiments} measures where this local temporal balance coincides with lower exploitability, and Appendix~\ref{app:scope} states how far the hypothesis is claimed to reach.

\section{Theoretical Analysis}
\label{sec:theory}

\paragraph{Variance Decomposition in MCCFR}

Throughout the analysis, MCCFR denotes the External Sampling variant adopted in this paper (Section~\ref{sec:mccfr-background}). As an interpretive approximation, we separate sampled-update variability into
\[
\widetilde{\mathrm{Var}} \approx V_{\mathrm{chance}} + V_{\mathrm{opp}}.
\]
CCS-MCCFR acts on the first term, and it does so through cumulative temporal covariance among chance contributions rather than by removing a fixed quantity: persistent low-discrepancy draws yield cancellation when downstream values vary slowly relative to the stream. The experiments locate the regimes in which this covariance effect is favorable, and Appendix~\ref{app:scope} states how the decomposition is used.

\paragraph{Temporal Covariance and Empirical Moderators}
\label{sec:revisit}

Two empirical features organize the effects across games. First, chance contributions must remain material after projection into regret differences, since symmetry can cancel them before temporal balancing matters. Second, the per-node streams need repeated visits for temporal coverage to differ meaningfully from isolated draws.

\textbf{The revisit diagnostic.} CCS-MCCFR spreads a node's outcomes across its repeated visits, so the length of a node's stream measures how much cross-visit structure the sampler has to work with. A node visited once offers none. In the tested set, Kuhn and Leduc have tens to hundreds of visits and show substantial gains, whereas a reduced Flop Hold'em has 75\% of chance nodes visited once and shows no gain (Table~\ref{tab:revisit}). Symmetric Goofspiel is revisited often and gains less, which is where the second feature enters.

\begin{observation}[Scope of Effectiveness]
\label{obs:scope}
In our tested tabular games, favorable effects are associated with material chance contributions to regret differences, repeated per-node visits, and private-information coupling. These are empirical correlates of the proposed temporal-cancellation mechanism (Appendix~\ref{app:scope}).
\end{observation}

\paragraph{Compatibility with Regret Update Rules}
\label{sec:theory-orth}

CCS-MCCFR modifies the chance draw while leaving the regret update algebra unchanged, so it can be implemented together with update rules such as Discounted CFR and Linear CFR \citep{brown2019solving}. The update rule does interact with the mechanism, because it changes the downstream weights paired with each phase and therefore the temporal covariance. Section~\ref{sec:orthogonality} measures this interaction on a composition grid in which CCS-MCCFR improves every tested update-rule configuration, Linear CFR included.

\paragraph{Game Dependent Effectiveness}

The observed effect may depend on:
\begin{itemize}
\item \textbf{Chance decision coupling:} When chance is interleaved with decisions and reveals private or public information, its temporally correlated contributions enter many downstream regret differences. Root-only chance also carries this coupling: Kuhn deals private cards once and still yields one of the largest measured gains.
\item \textbf{Outcome symmetry:} Symmetric payoff structure can cancel chance contributions in regret differences, leaving less variation for temporal balancing to affect.
\end{itemize}

\begin{observation}[Chance Decision Coupling]
\label{obs:coupling}
In the tested games, larger reductions are associated with repeated chance visits whose outcomes couple to downstream private-information decisions, while symmetry is associated with smaller effects (Appendix~\ref{app:scope}).
\end{observation}

\section{Experiments}
\label{sec:experiments}

\paragraph{Experimental Setup.} We evaluate vanilla i.i.d., antithetic, and CCS-MCCFR chance sampling on Kuhn, Leduc, Goofspiel-4, and Liar's Dice, using exact OpenSpiel exploitability and equal node-touch budgets \citep{lanctot2009monte,farina2020stochastic}. The four main games use 200 paired seeds; the controlled Leduc-family expansion uses 50. Significance decisions use 10{,}000-resample paired-bootstrap intervals on per-seed absolute differences, expressed relative to the vanilla mean, and an effect is detected when its 95\% interval clears zero. Appendix~\ref{app:experiment-protocol} gives the game definitions, sampler implementations, budgets, seed counts, ribbon conventions, and auxiliary paired tests.

\paragraph{Main Results}

Kuhn and Leduc exhibit 27.64\% and 24.59\% final exploitability reductions, respectively, with paired-bootstrap confidence intervals well above zero. Goofspiel-4 shows a smaller but significant 4.27\% reduction, whereas the Liar's Dice interval straddles zero. Figures~\ref{fig:convergence} and~\ref{fig:main-reduction} show the raw exploitability and paired-reduction trajectories, and Table~\ref{tab:main} reports the endpoints.

\begin{figure}[t]
\centering
\includegraphics[width=0.8\textwidth]{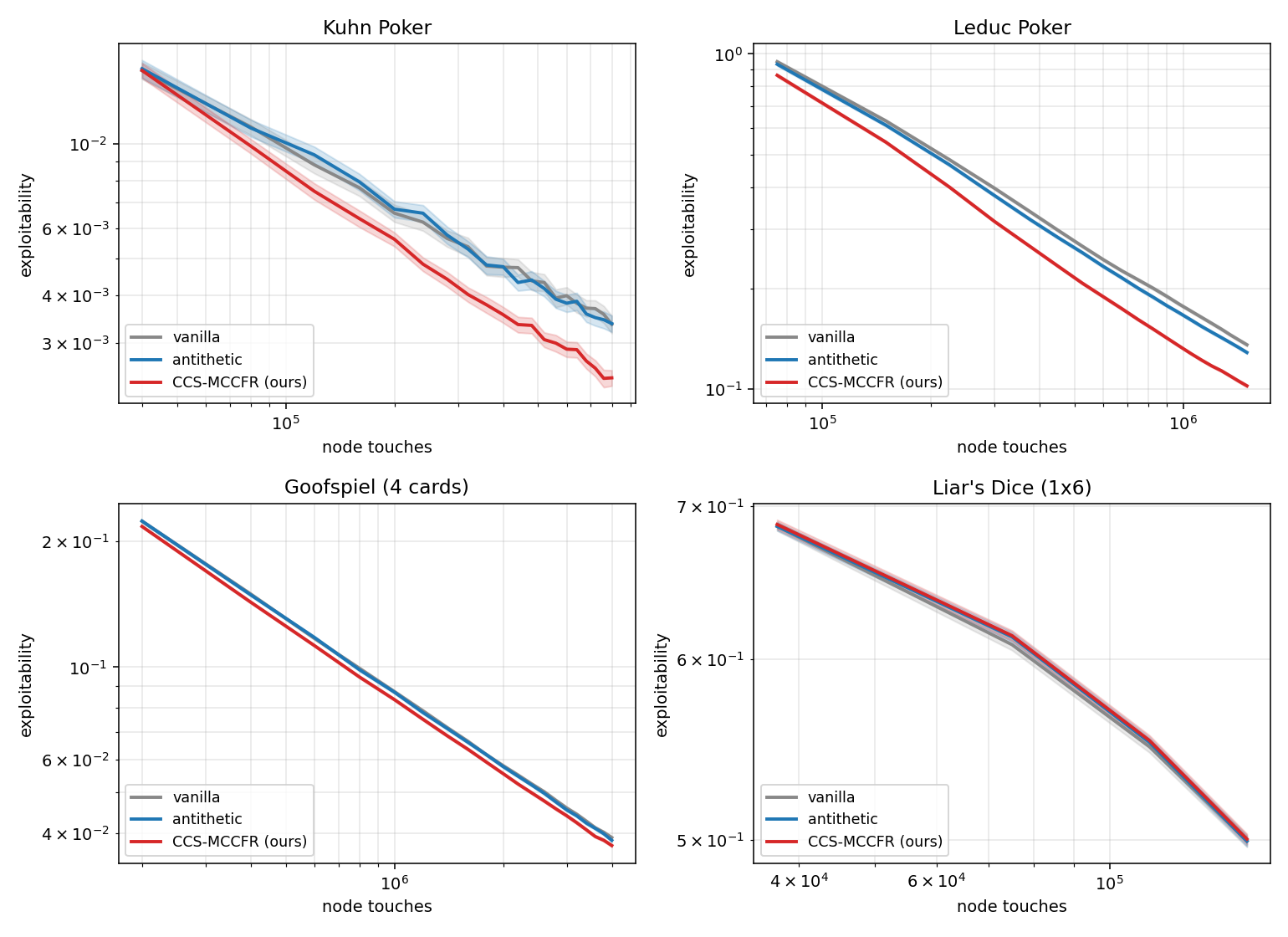}
\caption{Exploitability vs.\ node touches for four games. Each panel shows vanilla, antithetic, and CCS-MCCFR chance sampling; bands are 95\% confidence intervals (mean $\pm 1.96$\,SEM over $n{=}200$ paired seeds). CCS-MCCFR improves the poker benchmarks, yields a smaller detected effect on Goofspiel, and places the final Liar's Dice endpoint near vanilla with a paired reduction interval that includes zero.}
\label{fig:convergence}
\end{figure}

\begin{figure}[t]
\centering
\includegraphics[width=0.8\textwidth]{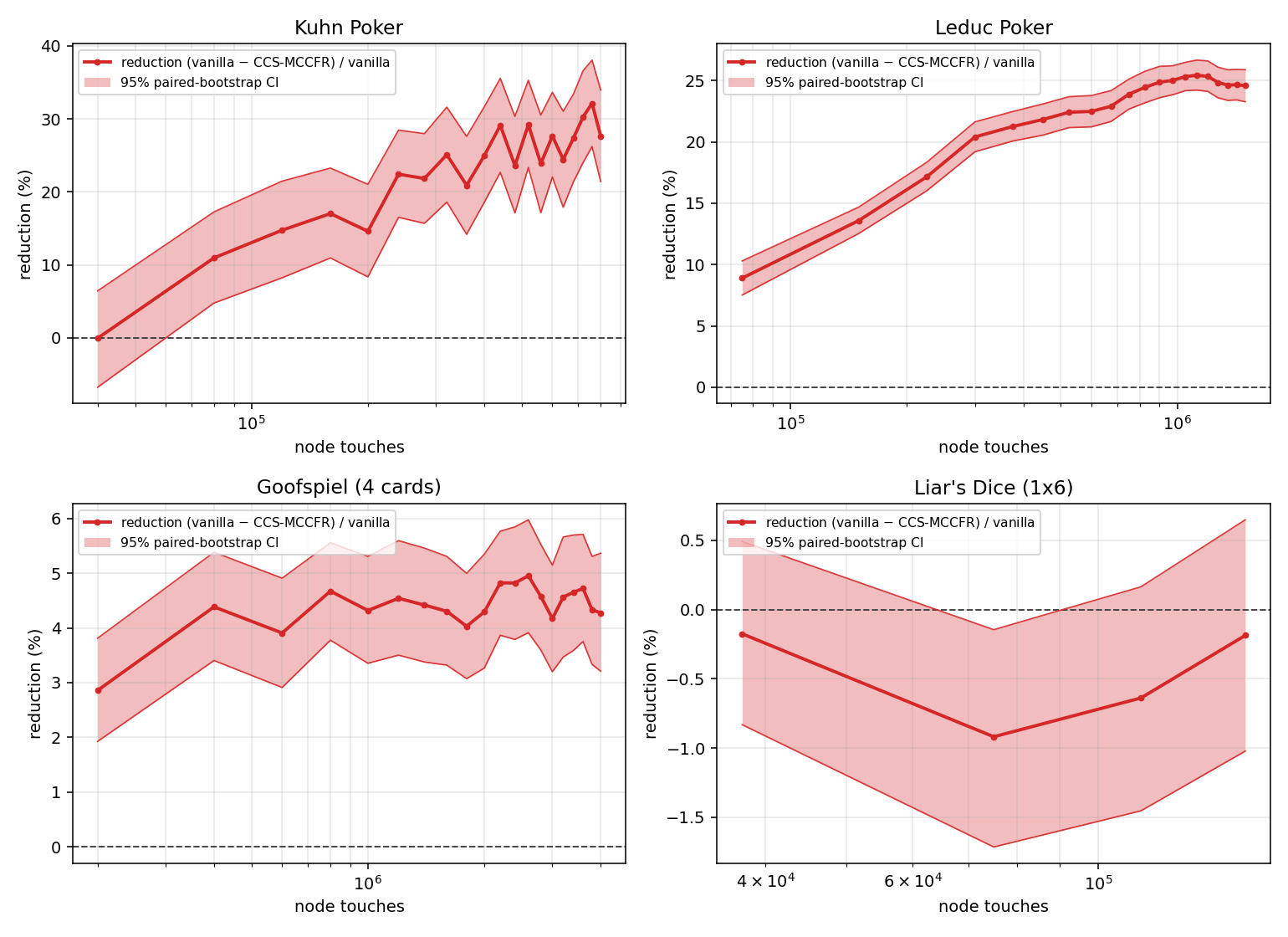}
\caption{Paired reduction $(\text{vanilla}-\text{CCS-MCCFR})/\text{vanilla}$ (\%) versus node touches for the four main games, with the 95\% paired-bootstrap CI band and a zero reference line ($n{=}200$ paired seeds). This uses the same statistic and convention as the HUNL reduction figures (Figure~\ref{fig:hunl}, Figure~\ref{fig:hunl-river}) and is the paired-reduction view of the raw exploitability ribbons in Figure~\ref{fig:convergence}. The band clears zero throughout for Kuhn, Leduc, and Goofspiel---a detected CCS-MCCFR gain that grows with the node touch budget---and straddles zero for Liar's Dice, matching the null in Table~\ref{tab:main}.}
\label{fig:main-reduction}
\end{figure}

\begin{table}[h]
\centering
\small
\caption{Final exploitability reduction vs.\ vanilla MCCFR, paired design ($n{=}200$ seeds for all four games). Exploitability is reported as mean $\pm 1.96$\,SEM over seeds. The CCS-MCCFR and Antithetic reduction columns are both relative to the same vanilla baseline. Significance is assessed using 10{,}000 resample paired bootstrap 95\% intervals; a result is significant iff its interval clears zero. In configurations with a detected CCS-MCCFR gain, the simple antithetic construction does not reproduce its magnitude.}
\label{tab:main}
\setlength{\tabcolsep}{4pt}
\resizebox{\textwidth}{!}{%
\begin{tabular}{lccccc}
\toprule
Game & Budget & Vanilla & CCS-MCCFR & CCS-MCCFR reduction & Antithetic reduction \\
\midrule
Kuhn        & 800k  & $0.00335\,{\pm}\,0.00018$ & $0.00243\,{\pm}\,0.00011$ & $\mathbf{27.64\%}$ & $-0.46\%$ \\
Leduc       & 1.5M  & $0.13563\,{\pm}\,0.00141$ & $0.10228\,{\pm}\,0.00100$ & $\mathbf{24.59\%}$ & $5.16\%$ \\
Goofspiel   & 4M    & $0.03902\,{\pm}\,0.00031$ & $0.03735\,{\pm}\,0.00030$ & $\mathbf{4.27\%}$  & $1.50\%$ \\
Liar's Dice & 150k  & $0.49944\,{\pm}\,0.00307$ & $0.50038\,{\pm}\,0.00268$ & $-0.19\%$ & $0.03\%$ \\
\bottomrule
\end{tabular}}
\end{table}

\textbf{Key findings:}
\begin{itemize}
\item \textbf{Poker games (Kuhn, Leduc):} CCS-MCCFR produces the largest main-benchmark reductions, 27.64\% and 24.59\%, with paired-bootstrap CIs well above zero. Both games repeatedly revisit private-information chance allocations under the node-touch budget; Leduc also interleaves chance and decisions.
\item \textbf{Goofspiel:} The 4.27\% reduction is smaller but significant ($4.3\%\,{\pm}\,1.1\%$, paired-bootstrap 95\% CI), extending the detected effect beyond poker while illustrating attenuation in a symmetric chance structure (Observation~\ref{obs:coupling}).
\item \textbf{Liar's Dice:} The 200-seed experiment precisely localizes an endpoint near zero: CCS-MCCFR and antithetic reductions are $-0.19\%$ and $0.03\%$, with both paired intervals crossing zero. This high-power boundary contrasts with Kuhn, which also places private chance at the root, and therefore rules out root-only placement as a sufficient explanation of the cross-game pattern.
\end{itemize}

\textbf{Antithetic control.} The antithetic column of Table~\ref{tab:main} evaluates a simpler paired correlation scheme, not CCS-MCCFR, across all four small games. It does not reproduce the magnitude of the CCS-MCCFR gains in the configurations where a gain is detected (e.g.\ Kuhn $-0.5\%$ vs.\ CCS-MCCFR $27.6\%$). On Liar's Dice, both constructions remain centered near the same vanilla endpoint.

\paragraph{Larger Poker Benchmarks: Scaling the Leduc Deck}
\label{sec:bigleduc}

The main experiments establish strong gains on standard Kuhn and Leduc. We next isolate deck size in a controlled Leduc-style limit-poker family implemented with OpenSpiel's \texttt{universal\_poker}: the structure remains fixed (one private hole card, one public board card, two betting rounds, fold/call betting), while the deck grows from 6 to 10 and 12 cards. This preserves the chance/decision and private-information structure while expanding the outcome space.

We use a paired design in which the same seed drives vanilla and CCS-MCCFR runs, with $n=50$ seeds and a 10{,}000 resample bootstrap confidence interval on the relative reduction. A result is positive when this interval clears zero.

\begin{table}[h]
\centering
\small
\caption{Larger Leduc family poker (universal\_poker, fold/call, 2 rounds, 30k node touches, $n=50$ paired). Exploitability is mean $\pm 1.96$\,SEM over seeds; reduction is the relative drop of the mean (positive $=$ improvement), reported as a point estimate. Significance is assessed using the paired bootstrap 95\% interval; all reductions are significant (interval clears zero; paired Cohen's $d=0.39$ to $0.60$).}
\label{tab:bigleduc}
\setlength{\tabcolsep}{4pt}
\begin{tabular}{lcccc}
\toprule
Deck & Vanilla & CCS-MCCFR & Reduction & two-sided paired $t$-test $p$ \\
\midrule
6 card (standard) & $0.00544\,{\pm}\,0.00090$ & $0.00393\,{\pm}\,0.00051$ & $\mathbf{27.65\%}$ & $0.008$ \\
10 card           & $0.00864\,{\pm}\,0.00120$ & $0.00570\,{\pm}\,0.00069$ & $\mathbf{34.01\%}$ & $0.0001$ \\
12 card           & $0.00823\,{\pm}\,0.00096$ & $0.00666\,{\pm}\,0.00059$ & $\mathbf{19.05\%}$ & $0.006$ \\
\bottomrule
\end{tabular}
\end{table}

The controlled deck configurations yield reductions of 27.65\%, 34.01\%, and 19.05\%, each with a bootstrap CI strictly above zero (Table~\ref{tab:bigleduc}; convergence curves in Figure~\ref{fig:bigleduc}). The result is stable across the tested 6/10/12-card expansion, with the 10-card configuration producing the largest measured reduction. This controlled family complements the standard Leduc result by showing that the benefit persists when the chance-outcome space is enlarged while the game structure is held fixed.

\begin{figure}[t]
\centering
\includegraphics[width=\textwidth]{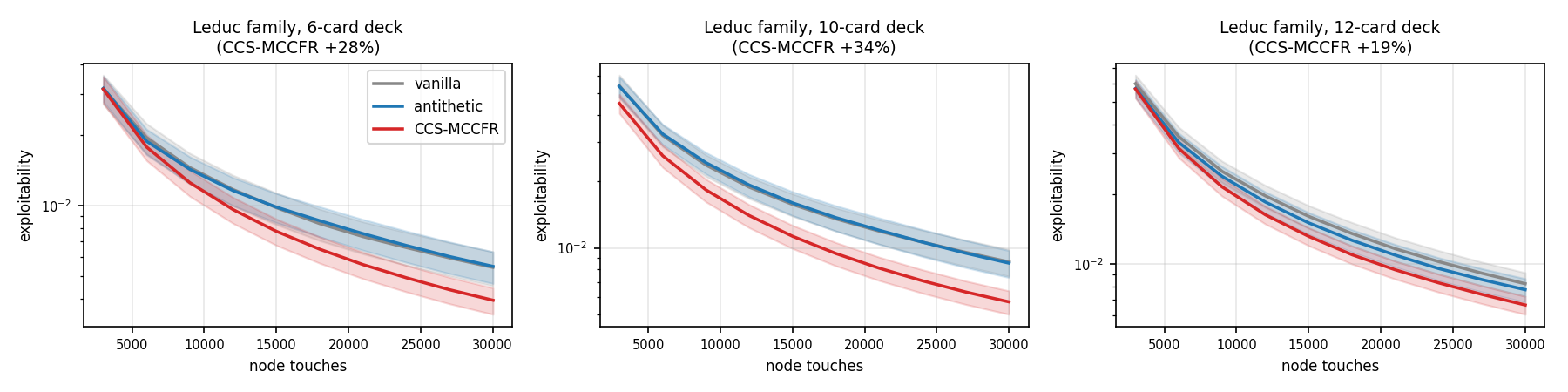}
\caption{Exploitability vs.\ node touches for the Leduc family poker games (6/10/12 card decks). Bands are 95\% confidence intervals (mean $\pm 1.96$\,SEM over $n{=}50$ paired seeds), using the same convention as Figure~\ref{fig:convergence}; since this figure uses 50 paired seeds whereas Figure~\ref{fig:convergence} uses 200, the absolute ribbon widths are not directly comparable across the two figures. The paired estimates favor CCS-MCCFR over vanilla throughout this 50 seed scaling experiment.}
\label{fig:bigleduc}
\end{figure}

\paragraph{An Empirical Revisit Diagnostic}
\label{sec:revisit-exp}

What distinguishes the strongest gain cases from the small effect and no gain cases? Our hypothesis suggests that repeated visits expose more cross-visit temporal structure for possible cancellation. Table~\ref{tab:revisit} reports this revisit diagnostic; Appendix~\ref{app:revisit-protocol} gives its measurement procedure.

The diagnostic separates the cases cleanly. Repeatedly revisited poker chance nodes benefit, the reduced Flop Hold'em configuration leaves most concrete streams nearly unused and shows no gain, and symmetric Goofspiel is revisited often yet gains less, which pairs exposure with symmetry as the two moderators.

At longer horizon, Leduc retains a 24.5\% reduction at 3M node touches (two-sided paired $t$-test $p=1.9\!\times\!10^{-6}$). A five-seed timing check at equal node-touch budget detects no solver-time difference (mean ratio $0.968$, paired $p=0.42$). Goofspiel-5 and full timing details are reported in Appendix~\ref{app:extra}.

\paragraph{Persistence Ablation}
\label{sec:persistence}

The persistence ablation separates the effective cross-iteration construction from a measured per-iteration reset: CCS-MCCFR improves both games, whereas the reset cell remains near vanilla. This supports across-iteration phase coupling as a contributor to the measured gains, while leaving the exact covariance pathway for future identification. A distinct per-traversal reset variant has the standard $O(1/\sqrt{T})$ External Sampling guarantee. The complete theorem, proof, implementation distinction, ablation table, and diagnostics appear in Appendix~\ref{app:reset-proof}.

\paragraph{Composition with Regret Update Accelerators}
\label{sec:orthogonality}

CCS-MCCFR can be implemented with different regret update rules because it changes the chance draw rather than the update formula. Table~\ref{tab:orthogonal} tests Discounted CFR and Linear CFR \citep{brown2019solving}, the two update rules that most later schedules build on \citep{xu2024dynamic,zhang2026faster}. CCS-MCCFR significantly reduces exploitability in every tested cell, with relative gains of 17.50\% to 24.30\%. LCFR plus CCS-MCCFR is the lowest of the six configurations and is 42.59\% below vanilla updates with i.i.d.\ chance. Appendix~\ref{app:composition-details} gives the update definitions, budget sweep, and implementation checks. A separate $2\times2$ experiment crosses the chance sampler with a Schmid-inspired opponent-action control variate restricted to External Sampling. The chance-sampling benefit survives with that control variate enabled, reaching 28.69\% on Kuhn and 25.64\% on Leduc; Appendix~\ref{app:es-control-variate} reports the design, the endpoints, and what the comparison establishes.

\begin{table}[h]
\centering
\small
\caption{Composition ablation: CCS-MCCFR $\times$ regret update rule on Leduc at $3\times10^5$ node touches (100 seeds, paired). Each chance cell reports exploitability (mean $\pm 1.96$\,SEM); the rightmost column is the paired CCS-MCCFR versus i.i.d.\ reduction point estimate. Significance is assessed using paired-bootstrap 95\% intervals. LCFR vs.\ vanilla update (i.i.d.\ chance): $+24.16\%$.}
\label{tab:orthogonal}
\begin{tabular}{lccc}
\toprule
Update rule & i.i.d.\ chance & CCS-MCCFR & reduction \\
\midrule
Vanilla CFR & $0.39539\,{\pm}\,0.00516$ & $0.31850\,{\pm}\,0.00381$ & \textbf{$+19.45\%$} \\
LCFR        & $0.29986\,{\pm}\,0.00540$ & $0.22699\,{\pm}\,0.00367$ & \textbf{$+24.30\%$} \\
DCFR        & $0.41196\,{\pm}\,0.00531$ & $0.33986\,{\pm}\,0.00406$ & \textbf{$+17.50\%$} \\
\bottomrule
\end{tabular}
\end{table}

\paragraph{HUNL transfer boundary.}
\label{sec:hunl}

We also test four fixed Libratus endgames with a standalone C++ External-Sampling MCCFR solver: Subgames~1--2 begin on the Turn and Subgames~3--4 on the River. Each retains a fixed action abstraction, uses no hand abstraction, and receives $2\times10^6$ iterations under a 25-seed paired i.i.d.-versus-Weyl protocol; exploitability is measured within the retained abstract endgame. The four final reduction point estimates are $+0.05\%$, $+0.00\%$, $+0.48\%$, and $+0.74\%$, and every paired interval crosses zero. These endgames therefore fall in the lower-exposure regime of the revisit diagnostic, with median visit counts of 3 to 10 against the tens to hundreds seen in tabular poker. Appendix~\ref{app:hunl-details} reports the full protocol, solver-revision scope, endpoint tables, and Turn/River trajectories; Table~\ref{tab:revisit} and Figure~\ref{fig:revisit} add the corresponding concrete-node exposure diagnostics.

\section{Conclusion}

CCS-MCCFR substantially reduces exploitability across tabular poker: 27.64\% on Kuhn, 24.59\% on standard Leduc, and 27.65\%, 34.01\%, and 19.05\% across a controlled 6/10/12-card Leduc family, with all paired-bootstrap intervals above zero. The Leduc benefit remains 24.5\% at 3M node touches, and combining CCS-MCCFR with Linear CFR produces the lowest measured composition cell, 42.59\% below vanilla updates with i.i.d.\ chance. Goofspiel-4 adds a smaller significant 4.27\% gain. Liar's Dice, the reduced Flop budget scan, Goofspiel-5, and four HUNL transfer stress tests place the endpoint at vanilla with no statistically detectable final-endpoint difference, and revisit exposure, symmetry, and private-information coupling organize the full pattern as empirical moderators.

The method achieves this as a minimal sampler change that leaves the chance law, the External-Sampling estimator, and the MCCFR regret updates intact. Its first $N$ draws at one concrete chance node, where $N$ is the node's consumed-draw count rather than global iterations, satisfy deterministic $O(\!\log(N+1)/N)$ frequency discrepancy, and fixed-index marginal correctness together with fixed-trajectory unbiasedness supply complementary local guarantees. The adaptive phase-selection bound then names the single dependence term that a global analysis must control, and per-traversal resetting already retains the standard $O(1/\sqrt{T})$ guarantee. The result is a sampler that costs nothing to adopt and delivers double-digit exploitability reductions wherever chance nodes are revisited under private information, and it opens adaptive correlated chance sampling as a distinct axis for accelerating equilibrium computation.

\section*{Acknowledgments}

The authors thank the anonymous reviewers for their helpful feedback.

\bibliographystyle{plainnat}
\bibliography{references}

\begin{thebibliography}{51}
\providecommand{\natexlab}[1]{#1}
\providecommand{\url}[1]{\texttt{#1}}
\expandafter\ifx\csname urlstyle\endcsname\relax
  \providecommand{\doi}[1]{doi: #1}\else
  \providecommand{\doi}{doi: \begingroup \urlstyle{rm}\Url}\fi

\bibitem[Avron et~al.(2016)Avron, Sindhwani, Yang, and Mahoney]{avron2016quasi}
Haim Avron, Vikas Sindhwani, Jiyan Yang, and Michael~W Mahoney.
\newblock Quasi-monte carlo feature maps for shift-invariant kernels.
\newblock \emph{Journal of Machine Learning Research}, 17\penalty0
  (120):\penalty0 1--38, 2016.

\bibitem[Baghal(2025)]{baghal2025pasur}
Sina Baghal.
\newblock Solving pasur using gpu-accelerated counterfactual regret
  minimization.
\newblock \emph{arXiv preprint arXiv:2508.06559}, 2025.

\bibitem[Bertram et~al.(2024)Bertram, F{\"u}rnkranz, and
  M{\"u}ller]{bertram2024efficiently}
Timo Bertram, Johannes F{\"u}rnkranz, and Martin M{\"u}ller.
\newblock Efficiently training neural networks for imperfect information games
  by sampling information sets.
\newblock In \emph{German Conference on Artificial Intelligence (K{\"u}nstliche
  Intelligenz)}, pages 17--29. Springer, 2024.

\bibitem[Brown and Sandholm(2018)]{brown2018superhuman}
Noam Brown and Tuomas Sandholm.
\newblock Superhuman ai for heads-up no-limit poker: Libratus beats top
  professionals.
\newblock \emph{Science}, 359\penalty0 (6374):\penalty0 418--424, 2018.

\bibitem[Brown and Sandholm(2019{\natexlab{a}})]{brown2019solving}
Noam Brown and Tuomas Sandholm.
\newblock Solving imperfect-information games via discounted regret
  minimization.
\newblock In \emph{Proceedings of the AAAI Conference on Artificial
  Intelligence}, volume~33, pages 1829--1836, 2019{\natexlab{a}}.

\bibitem[Brown and Sandholm(2019{\natexlab{b}})]{brown2019superhuman}
Noam Brown and Tuomas Sandholm.
\newblock Superhuman ai for multiplayer poker.
\newblock \emph{Science}, 365\penalty0 (6456):\penalty0 885--890,
  2019{\natexlab{b}}.

\bibitem[Brown et~al.(2019)Brown, Lerer, Gross, and Sandholm]{brown2019deep}
Noam Brown, Adam Lerer, Sam Gross, and Tuomas Sandholm.
\newblock Deep counterfactual regret minimization.
\newblock In \emph{International conference on machine learning}, pages
  793--802. PMLR, 2019.

\bibitem[Brown et~al.(2020)Brown, Bakhtin, Lerer, and Gong]{brown2020combining}
Noam Brown, Anton Bakhtin, Adam Lerer, and Qucheng Gong.
\newblock Combining deep reinforcement learning and search for
  imperfect-information games.
\newblock \emph{Advances in neural information processing systems},
  33:\penalty0 17057--17069, 2020.

\bibitem[Buchholz et~al.(2018)Buchholz, Wenzel, and Mandt]{buchholz2018quasi}
Alexander Buchholz, Florian Wenzel, and Stephan Mandt.
\newblock Quasi-monte carlo variational inference.
\newblock In \emph{Proceedings of the 35th International Conference on Machine
  Learning}, volume~80 of \emph{Proceedings of Machine Learning Research},
  pages 667--676. PMLR, 2018.

\bibitem[Burch et~al.(2012)Burch, Lanctot, Szafron, and
  Gibson]{burch2012efficient}
Neil Burch, Marc Lanctot, Duane Szafron, and Richard Gibson.
\newblock Efficient monte carlo counterfactual regret minimization in games
  with many player actions.
\newblock \emph{Advances in neural information processing systems}, 25, 2012.

\bibitem[Cranley and Patterson(1976)]{cranley1976randomization}
Roy Cranley and Thomas~NL Patterson.
\newblock Randomization of number theoretic methods for multiple integration.
\newblock \emph{SIAM Journal on Numerical Analysis}, 13\penalty0 (6):\penalty0
  904--914, 1976.

\bibitem[Davis et~al.(2020)Davis, Schmid, and Bowling]{davis2020low}
Trevor Davis, Martin Schmid, and Michael Bowling.
\newblock Low-variance and zero-variance baselines for extensive-form games.
\newblock In \emph{International Conference on Machine Learning}, pages
  2392--2401. PMLR, 2020.

\bibitem[Farina et~al.(2020)Farina, Kroer, and Sandholm]{farina2020stochastic}
Gabriele Farina, Christian Kroer, and Tuomas Sandholm.
\newblock Stochastic regret minimization in extensive-form games.
\newblock In \emph{International Conference on Machine Learning}, pages
  3018--3028. PMLR, 2020.

\bibitem[Farina et~al.(2021)Farina, Kroer, and Sandholm]{farina2021faster}
Gabriele Farina, Christian Kroer, and Tuomas Sandholm.
\newblock Faster game solving via predictive blackwell approachability:
  Connecting regret matching and mirror descent.
\newblock In \emph{Proceedings of the AAAI Conference on Artificial
  Intelligence}, volume~35, pages 5363--5371, 2021.

\bibitem[Gibson et~al.(2012)Gibson, Lanctot, Burch, Szafron, and
  Bowling]{gibson2012generalized}
Richard Gibson, Marc Lanctot, Neil Burch, Duane Szafron, and Michael Bowling.
\newblock Generalized sampling and variance in counterfactual regret
  minimization.
\newblock In \emph{Proceedings of the AAAI Conference on Artificial
  Intelligence}, volume~26, pages 1355--1361, 2012.

\bibitem[Halton(1960)]{halton1960efficiency}
John~H Halton.
\newblock On the efficiency of certain quasi-random sequences of points in
  evaluating multi-dimensional integrals.
\newblock \emph{Numerische Mathematik}, 2\penalty0 (1):\penalty0 84--90, 1960.

\bibitem[Johanson et~al.(2012)Johanson, Bard, Lanctot, Gibson, and
  Bowling]{johanson2012efficient}
Michael Johanson, Nolan Bard, Marc Lanctot, Richard~G Gibson, and Michael
  Bowling.
\newblock Efficient nash equilibrium approximation through monte carlo
  counterfactual regret minimization.
\newblock In \emph{AAMAS}, pages 837--846, 2012.

\bibitem[Kova{\v{r}}{\'\i}k et~al.(2022)Kova{\v{r}}{\'\i}k, Schmid, Burch,
  Bowling, and Lis{\'y}]{kovarik2022rethinking}
Vojt{\v{e}}ch Kova{\v{r}}{\'\i}k, Martin Schmid, Neil Burch, Michael Bowling,
  and Viliam Lis{\'y}.
\newblock Rethinking formal models of partially observable multiagent decision
  making.
\newblock \emph{Artificial Intelligence}, 303:\penalty0 103645, 2022.

\bibitem[Lanctot et~al.(2009)Lanctot, Waugh, Zinkevich, and
  Bowling]{lanctot2009monte}
Marc Lanctot, Kevin Waugh, Martin Zinkevich, and Michael Bowling.
\newblock Monte carlo sampling for regret minimization in extensive games.
\newblock \emph{Advances in neural information processing systems}, 22, 2009.

\bibitem[Li and Huang(2025)]{li2025efficient}
Boning Li and Longbo Huang.
\newblock Efficient online pruning and abstraction for imperfect information
  extensive-form games.
\newblock In \emph{The Thirteenth International Conference on Learning
  Representations}, 2025.

\bibitem[Li and Huang(2026{\natexlab{a}})]{li2026effective}
Boning Li and Longbo Huang.
\newblock Effective, efficient, and general information abstraction for
  imperfect-information extensive-form games.
\newblock \emph{arXiv preprint arXiv:2605.10900}, 2026{\natexlab{a}}.

\bibitem[Li and Huang(2026{\natexlab{b}})]{li2026parallel}
Boning Li and Longbo Huang.
\newblock Real-time parallel counterfactual regret minimization.
\newblock \emph{arXiv preprint arXiv:2605.19928}, 2026{\natexlab{b}}.

\bibitem[Li et~al.(2024)Li, Fang, and Huang]{li2024rl}
Boning Li, Zhixuan Fang, and Longbo Huang.
\newblock Rl-cfr: improving action abstraction for imperfect information
  extensive-form games with reinforcement learning.
\newblock In \emph{Proceedings of the 41st International Conference on Machine
  Learning}, pages 27752--27770, 2024.

\bibitem[Li et~al.(2026)Li, Wang, and Huang]{li2026pokerskill}
Boning Li, Baoxiang Wang, and Longbo Huang.
\newblock Pokerskill: Llms can play expert-level poker without training or
  solvers.
\newblock \emph{arXiv preprint arXiv:2605.30094}, 2026.

\bibitem[Lis{\'y} et~al.(2015)Lis{\'y}, Lanctot, and Bowling]{lisy2015online}
Viliam Lis{\'y}, Marc Lanctot, and Michael~H Bowling.
\newblock Online monte carlo counterfactual regret minimization for search in
  imperfect information games.
\newblock In \emph{AAMAS}, pages 27--36, 2015.

\bibitem[Madow and Madow(1944)]{madow1944theory}
William~G Madow and Lillian~H Madow.
\newblock On the theory of systematic sampling, i.
\newblock \emph{The Annals of Mathematical Statistics}, 15\penalty0
  (1):\penalty0 1--24, 1944.

\bibitem[Mania et~al.(2018)Mania, Guy, and Recht]{mania2018simple}
Horia Mania, Aurelia Guy, and Benjamin Recht.
\newblock Simple random search of static linear policies is competitive for
  reinforcement learning.
\newblock \emph{Advances in neural information processing systems}, 31, 2018.

\bibitem[McAleer et~al.(2023)McAleer, Farina, Lanctot, and
  Sandholm]{mcaleer2023escher}
Stephen~Marcus McAleer, Gabriele Farina, Marc Lanctot, and Tuomas Sandholm.
\newblock Escher: Eschewing importance sampling in games by computing a history
  value function to estimate regret.
\newblock In \emph{The Eleventh International Conference on Learning
  Representations}, 2023.

\bibitem[Meng et~al.(2025)Meng, Chen, Li, Yang, Zhang, and
  Gao]{meng2025reducing}
Linjian Meng, Wubing Chen, Wenbin Li, Tianpei Yang, Youzhi Zhang, and Yang Gao.
\newblock Reducing variance of stochastic optimization for approximating nash
  equilibria in normal-form games.
\newblock In \emph{Forty-second International Conference on Machine Learning},
  2025.

\bibitem[Meng et~al.(2026{\natexlab{a}})Meng, Yang, Zhang, Ge, and
  Gao]{meng2026faster}
Linjian Meng, Tianpei Yang, Youzhi Zhang, Zhenxing Ge, and Yang Gao.
\newblock Faster game solving via asymmetry of step sizes.
\newblock In \emph{Proceedings of the AAAI Conference on Artificial
  Intelligence}, volume~40, pages 17161--17169, 2026{\natexlab{a}}.

\bibitem[Meng et~al.(2026{\natexlab{b}})Meng, Zhang, Yang, Li, Ding, and
  Gao]{meng2026parameterfree}
Linjian Meng, Youzhi Zhang, Shangdong Yang, Wenbin Li, Tianyu Ding, and Yang
  Gao.
\newblock A faster parameter-free regret matching algorithm.
\newblock In \emph{The Fourteenth International Conference on Learning
  Representations}, 2026{\natexlab{b}}.

\bibitem[Mohamed et~al.(2020)Mohamed, Rosca, Figurnov, and
  Mnih]{mohamed2020monte}
Shakir Mohamed, Mihaela Rosca, Michael Figurnov, and Andriy Mnih.
\newblock Monte carlo gradient estimation in machine learning.
\newblock \emph{Journal of Machine Learning Research}, 21\penalty0
  (132):\penalty0 1--62, 2020.

\bibitem[Morav{\v{c}}{\'\i}k et~al.(2017)Morav{\v{c}}{\'\i}k, Schmid, Burch,
  Lis{\'y}, Morrill, Bard, Davis, Waugh, Johanson, and
  Bowling]{moravcik2017deepstack}
Matej Morav{\v{c}}{\'\i}k, Martin Schmid, Neil Burch, Viliam Lis{\'y}, Dustin
  Morrill, Nolan Bard, Trevor Davis, Kevin Waugh, Michael Johanson, and Michael
  Bowling.
\newblock {DeepStack}: Expert-level artificial intelligence in heads-up
  no-limit poker.
\newblock \emph{Science}, 356\penalty0 (6337):\penalty0 508--513, 2017.

\bibitem[Niederreiter(1992)]{niederreiter1992random}
Harald Niederreiter.
\newblock \emph{Random number generation and quasi-Monte Carlo methods}.
\newblock SIAM, 1992.

\bibitem[Owen(1997)]{owen1997monte}
Art~B Owen.
\newblock Monte carlo variance of scrambled net quadrature.
\newblock \emph{SIAM Journal on Numerical Analysis}, 34\penalty0 (5):\penalty0
  1884--1910, 1997.

\bibitem[Peshkin and Shelton(2002)]{peshkin2002learning}
Leonid Peshkin and Christian~R Shelton.
\newblock Learning from scarce experience.
\newblock In \emph{Proceedings of the Nineteenth International Conference on
  Machine Learning}, pages 498--505, 2002.

\bibitem[Qi et~al.(2024)Qi, Falin, Ting, Dengbing, Zhemei, and
  Yunfeng]{qi2024accelerating}
Ju~Qi, Hei Falin, Feng Ting, Yi~Dengbing, Fang Zhemei, and Luo Yunfeng.
\newblock Accelerating nash equilibrium convergence in monte carlo settings
  through counterfactual value based fictitious play.
\newblock \emph{Advances in Neural Information Processing Systems},
  37:\penalty0 108088--108115, 2024.

\bibitem[Schmid et~al.(2019)Schmid, Burch, Lanctot, Moravcik, Kadlec, and
  Bowling]{schmid2019variance}
Martin Schmid, Neil Burch, Marc Lanctot, Matej Moravcik, Rudolf Kadlec, and
  Michael Bowling.
\newblock Variance reduction in monte carlo counterfactual regret minimization
  (vr-mccfr) for extensive form games using baselines.
\newblock In \emph{Proceedings of the AAAI Conference on Artificial
  Intelligence}, volume~33, pages 2157--2164, 2019.

\bibitem[Schmid et~al.(2023)Schmid, Morav{\v{c}}{\'\i}k, Burch, Kadlec,
  Davidson, Waugh, Bard, Timbers, Lanctot, Holland, et~al.]{schmid2023student}
Martin Schmid, Matej Morav{\v{c}}{\'\i}k, Neil Burch, Rudolf Kadlec, Josh
  Davidson, Kevin Waugh, Nolan Bard, Finbarr Timbers, Marc Lanctot, G~Zacharias
  Holland, et~al.
\newblock Student of games: A unified learning algorithm for both perfect and
  imperfect information games.
\newblock \emph{Science Advances}, 9\penalty0 (46):\penalty0 eadg3256, 2023.

\bibitem[Sobol(1967)]{sobol1967distribution}
Ilya~M Sobol.
\newblock Distribution of points in a cube and approximate evaluation of
  integrals.
\newblock \emph{USSR Computational mathematics and mathematical physics},
  7:\penalty0 86--112, 1967.

\bibitem[Srinivasan et~al.(2018)Srinivasan, Lanctot, Zambaldi, P{\'e}rolat,
  Tuyls, Munos, and Bowling]{srinivasan2018actor}
Sriram Srinivasan, Marc Lanctot, Vin{\'i}cius~Flores Zambaldi, Julien
  P{\'e}rolat, Karl Tuyls, R{\'e}mi Munos, and Michael Bowling.
\newblock Actor-critic policy optimization in partially observable multiagent
  environments.
\newblock In \emph{Advances in Neural Information Processing Systems 31
  (NeurIPS 2018)}, pages 3426--3439, 2018.

\bibitem[Steinberger et~al.(2020)Steinberger, Lerer, and
  Brown]{steinberger2020dream}
Eric Steinberger, Adam Lerer, and Noam Brown.
\newblock Dream: Deep regret minimization with advantage baselines and
  model-free learning.
\newblock \emph{arXiv preprint arXiv:2006.10410}, 2020.

\bibitem[{\v{S}}ustr et~al.(2019){\v{S}}ustr, Kova{\v{r}}{\'\i}k, and
  Lis{\'y}]{sustr2019monte}
Michal {\v{S}}ustr, Vojt{\v{e}}ch Kova{\v{r}}{\'\i}k, and Viliam Lis{\'y}.
\newblock Monte carlo continual resolving for online strategy computation in
  imperfect information games.
\newblock In \emph{Proceedings of the 18th International Conference on
  Autonomous Agents and MultiAgent Systems}, pages 224--232, 2019.

\bibitem[{\v{S}}ustr et~al.(2020){\v{S}}ustr, Schmid, Morav{\v{c}}{\'\i}k,
  Burch, Lanctot, and Bowling]{sustr2020sound}
Michal {\v{S}}ustr, Martin Schmid, Matej Morav{\v{c}}{\'\i}k, Neil Burch, Marc
  Lanctot, and Michael Bowling.
\newblock Sound search in imperfect information games.
\newblock \emph{arXiv preprint arXiv:2006.08740}, 2020.

\bibitem[Sychrovsk{\'y} et~al.(2025)Sychrovsk{\'y}, Schmid, Sustr, and
  Bowling]{sychrovsky2025meta}
David Sychrovsk{\'y}, Martin Schmid, Michal Sustr, and Michael Bowling.
\newblock Meta-learning in self-play regret minimization.
\newblock \emph{arXiv preprint arXiv:2504.18917}, 2025.

\bibitem[Tammelin(2014)]{tammelin2014solving}
Oskari Tammelin.
\newblock Solving large imperfect information games using cfr+.
\newblock \emph{arXiv preprint arXiv:1407.5042}, 2014.

\bibitem[Waugh et~al.(2015)Waugh, Morrill, Bagnell, and
  Bowling]{waugh2015solving}
Kevin Waugh, Dustin Morrill, James Bagnell, and Michael Bowling.
\newblock Solving games with functional regret estimation.
\newblock In \emph{Proceedings of the AAAI Conference on Artificial
  Intelligence}, volume~29, 2015.

\bibitem[Xu et~al.(2024)Xu, Li, Fu, Fu, Xing, and Cheng]{xu2024dynamic}
Hang Xu, Kai Li, Haobo Fu, Qiang Fu, Junliang Xing, and Jian Cheng.
\newblock Dynamic discounted counterfactual regret minimization.
\newblock In \emph{The Twelfth International Conference on Learning
  Representations}, 2024.

\bibitem[Xu et~al.(2026)Xu, Li, Fu, Fu, Xing, and Cheng]{xu2026deep}
Hang Xu, Kai Li, Haobo Fu, Qiang Fu, Junliang Xing, and Jian Cheng.
\newblock Deep (predictive) discounted counterfactual regret minimization.
\newblock In \emph{Proceedings of the AAAI Conference on Artificial
  Intelligence}, volume~40, pages 17284--17292, 2026.

\bibitem[Zhang et~al.(2026)Zhang, McAleer, and Sandholm]{zhang2026faster}
Naifeng Zhang, Stephen~Marcus McAleer, and Tuomas Sandholm.
\newblock Faster game solving via hyperparameter schedules.
\newblock In \emph{Proceedings of the AAAI Conference on Artificial
  Intelligence}, volume~40, pages 17319--17326, 2026.

\bibitem[Zinkevich et~al.(2007)Zinkevich, Johanson, Bowling, and
  Piccione]{zinkevich2007regret}
Martin Zinkevich, Michael Johanson, Michael Bowling, and Carmelo Piccione.
\newblock Regret minimization in games with incomplete information.
\newblock \emph{Advances in neural information processing systems}, 20, 2007.

\end{thebibliography}

\appendix

\section{Scope of Claims}
\label{app:scope}

This appendix collects, in one place, the intended scope of the paper's claims.

\textbf{Empirical moderators.} Revisit exposure, outcome symmetry, and private-information coupling organize the measured results within this benchmark set. They are empirical correlates of the temporal-cancellation mechanism rather than established necessary or sufficient conditions, and Observations~\ref{obs:scope} and~\ref{obs:coupling} are descriptive summaries rather than theorem-level predictors. Revisit count is an exposure diagnostic: longer per-node streams make cross-visit cancellation possible without guaranteeing it, because the downstream weights are adaptive. Symmetric Goofspiel is revisited often and still gains less, so high revisit count alone is not sufficient; Kuhn places private chance only at the root and still gains substantially, so placement alone is not decisive either.

\textbf{Variance decomposition.} The split $\widetilde{\mathrm{Var}} \approx V_{\mathrm{chance}} + V_{\mathrm{opp}}$ in Section~\ref{sec:theory} is an interpretive approximation. It does not imply that CCS-MCCFR subtracts a nonnegative $V_{\mathrm{chance}}$ term. Under arbitrary adaptive bounded weights, the cumulative temporal covariance among chance contributions can vanish or take either sign, so neither the sign nor the magnitude of the effect is guaranteed by the decomposition.

\textbf{Boundary configurations.} Where the paper reports no statistically detectable final-endpoint difference, the statement means that the paired-bootstrap interval crosses zero at the stated budget and seed count. It is not a formal non-inferiority claim, and the underlying point estimates are not uniformly nonnegative: Liar's Dice is $-0.19\%$ at 200 seeds, the reduced Flop budget scan is $-2.81\%$, $-0.38\%$, and $+2.32\%$ at 8k, 15k, and 50k node touches, and Goofspiel-5 is $-0.8\%$ at 6 seeds. The paper does not attribute these boundary cases to a single cause.

\textbf{Game constructions.} The reduced Flop configuration is a targeted low-exposure diagnostic whose structure differs from full-scale Flop Hold'em, and its point estimates should not be read as claims about the latter. The four HUNL endgames are a protocol-mixed transfer diagnostic: exploitability is measured within each retained abstract endgame rather than in the full game, and Appendix~\ref{app:hunl-details} records the solver-revision scope.

\textbf{Composition.} The update-rule grid and the External-Sampling control-variate design establish empirical compatibility across the tested configurations. Neither supports an invariant or multiplicative composition law: the update rule changes the downstream weights paired with each phase and can therefore alter temporal covariance, and the measured reductions differ across cells. The control variate of Appendix~\ref{app:es-control-variate} is restricted to the opponent-action sample in External Sampling and is not the Outcome-Sampling instantiation of full VR-MCCFR.

\textbf{Theoretical reach.} Proposition~\ref{prop:biasbound} controls one bounded scalar contribution at one node and time through the conditional quantity $\delta_{c,t}$. This quantity is distinct from the frequency discrepancy of Theorem~\ref{thm:freq}, ordinary outcome-count logs need not identify it, and the inequality can be loose when $\delta_{c,t}$ is near one. Extending it across nodes, information sets, and adaptive iterations requires additional control of dependence and accumulation. Theorem~\ref{thm:freq} likewise controls local unweighted counts rather than adaptive weighted regret. Proposition~\ref{prop:traj} resolves the phase-selection term along fixed trajectories and Theorem~\ref{thm:reset} gives a per-traversal reset variant with the standard External Sampling guarantee, while global convergence of fully adaptive CCS-MCCFR remains open; Appendix~\ref{app:future} states the corresponding program.

\section{Future Work}
\label{app:future}

Two directions follow directly from the results. First, the controlled Leduc family shows that persistent balancing remains valuable as the chance-outcome space grows, while the HUNL endgames locate a lower-exposure regime at the same scale. Measuring cross-visit exposure and downstream coupling in larger games would turn these moderators into prospective predictors of where persistent streams retain enough temporal structure to help, and would extend the analysis to full-scale Flop Hold'em once a scalable best-response evaluator is available.

Second, Proposition~\ref{prop:biasbound} isolates adaptive phase selection through $\delta_{c,t}$. Establishing conditions that make this quantity decay, and controlling how local terms accumulate into global regret, would connect the per-node discrepancy result to a convergence guarantee for CCS-MCCFR. Randomized digital nets or lattice constructions may then extend the same approach to higher-dimensional chance spaces while preserving marginal correctness. The per-traversal reset variant already provides a standard-guarantee baseline for this theoretical program.

\section{Experimental Protocol Details}
\label{app:experiment-protocol}

The four principal games are Kuhn poker (root private-card chance), standard six-card Leduc poker (private cards and a public card between betting rounds), Goofspiel-4 (sequential symmetric public point-card reveals), and $1{\times}6$ Liar's Dice (root private dice followed by bidding). Vanilla sampling draws independently through the chance distribution, antithetic sampling pairs phases $u$ and $1-u$, and CCS-MCCFR uses the persistent randomized Weyl construction in Section~\ref{sec:method}. Exact exploitability is computed with OpenSpiel, and budgets count all decision and chance-node touches. The four principal experiments use 200 paired seeds; the controlled Leduc family uses 50, the update-rule grid 100, persistence ablations 10, reduced Flop 20, and Goofspiel-5 6. Raw exploitability ribbons are mean $\pm1.96$ SEM, whereas paired-reduction ribbons and significance decisions use 10{,}000-resample paired-bootstrap intervals on the per-seed absolute difference, normalized by the vanilla mean. Explicit $p$-values are complementary two-sided paired $t$-tests unless stated otherwise. The HUNL protocol and metric are detailed separately in Appendix~\ref{app:hunl-details}.

\section{Random Reset Guarantee}
\label{app:reset-proof}

\begin{theorem}[Convergence of the randomly reset variant]
\label{thm:reset}
Consider External Sampling MCCFR in a two player zero sum game where, at the start of every traversal, each concrete chance node $c$ independently resets its phase $\phi_c\sim\mathrm{Unif}[0,1)$, and chance outcomes within the traversal are drawn from the node's Weyl stream as in Section~\ref{sec:method}. Then the sampled traversals have the same joint distribution as under independent chance sampling; consequently, for each player $i$ and any $\rho\in(0,1]$, with probability at least $1-\rho$ the average regret after $T$ iterations satisfies
\[
\frac{R_i^T}{T} \;\le\; \Big(1+\frac{\sqrt{2}}{\sqrt{\rho}}\Big)\,
\frac{\Delta_{u,i}\,|\mathcal{I}_i|\,\sqrt{|A_i|}}{\sqrt{T}},
\]
the External Sampling bound of \citet{lanctot2009monte} (with their constant $M_i$ coarsened to $|\mathcal{I}_i|$), where $\Delta_{u,i}$ is the range of player $i$'s utilities, $\mathcal{I}_i$ its information sets, and $|A_i|=\max_{I\in\mathcal{I}_i}|A(I)|$. In particular the average profile converges to a Nash equilibrium at rate $O(1/\sqrt{T})$.
\end{theorem}

\begin{proof}
Fix a traversal and condition on the entire past (all previous draws and regret updates) and on the current strategy profile. The freshly reset phases are independent of this conditioning by construction. In an acyclic tree each concrete chance node is consumed at most once during a single traversal, so a draw at $c$, if it occurs, uses index $n=0$ and phase $u_{c,0}=\phi_c$, which is uniform; by Proposition~\ref{prop:unbiased} the outcome has law $f_c$. Whether the traversal reaches $c$ is determined by the strategy profile and by draws at other nodes and at $c$'s ancestors, all independent of $\phi_c$; and phases at distinct nodes are independent. Hence, conditional on the past, the joint law of the traversal's sampled chance outcomes, together with the opponent's sampled actions (drawn exactly as in standard External Sampling), coincides with the joint law under independent chance sampling. The bound of \citet{lanctot2009monte} for External Sampling is a statement about precisely this per traversal sampling distribution, composed over iterations, so it applies verbatim; the standard argument converts the regret bound into $O(1/\sqrt{T})$ convergence of the average profile.
\end{proof}

\begin{table}[h]
\centering
\small
\caption{Persistence ablation using the per-iteration reset implementation (10 seeds, 30k iterations). Exploitability is mean $\pm 1.96$\,SEM over seeds. The reset cell is statistically indistinguishable from vanilla while the CCS-MCCFR cell is lower, making across-iteration phase dependence a plausible contributor to the observed effect. Theorem~\ref{thm:reset} covers a distinct per traversal reset variant that was not separately measured in this table.}
\label{tab:persistence}
\begin{tabular}{lccc}
\toprule
Variant & Kuhn & Leduc & vs.\ vanilla \\
\midrule
Vanilla              & $0.00503\,{\pm}\,0.00116$ & $0.13889\,{\pm}\,0.00903$ & N/A \\
CCS-MCCFR (persistent) & $0.00362\,{\pm}\,0.00063$ & $0.10909\,{\pm}\,0.00353$ & \textbf{sig.} \\
Randomly reset        & $0.00554\,{\pm}\,0.00069$ & $0.14652\,{\pm}\,0.00637$ & no detected diff. \\
\bottomrule
\multicolumn{4}{l}{\footnotesize Persist vs.\ random reset (two-sided paired $t$-test): Kuhn +34.7\% ($p{=}.0043$), Leduc +25.5\% ($p{<}.0001$).}
\end{tabular}
\end{table}

\textbf{Implementation granularity.} Our implementation resets once per CFR iteration (covering both players' traversals) rather than once per traversal. The two coincide unless both traversals of the same iteration consume the same concrete node; in that case, under alternating regret updates, the second traversal's visit indicator can in principle covary with $\phi_c$ through the intervening update, a single step, within iteration instance of the phase selection effect of Proposition~\ref{prop:biasbound}. Resetting per traversal removes even this and is the variant covered by Theorem~\ref{thm:reset}; the measured per iteration reset cell is statistically indistinguishable from vanilla in Table~\ref{tab:persistence}.

\textbf{Persistence ablation and diagnostics.} The per-iteration reset implementation remains near vanilla in this ablation, while CCS-MCCFR improves on both, supporting across-iteration phase dependence as a contributor to the observed effect. This ablation does not identify the entire covariance pathway. For CCS-MCCFR, Proposition~\ref{prop:biasbound} identifies the relevant local conditional scalar term; Appendix~\ref{app:adaptive-frequency} reports complementary adaptive-frequency and large-budget diagnostics.

\section{Revisit Predictor Measurement}
\label{app:revisit-protocol}

We measure the revisit diagnostic by running vanilla MCCFR for 1000 iterations on each game and recording visits to the concrete chance keys actually observed. Table~\ref{tab:revisit} reports the median count and the fraction of observed keys visited exactly once. The four HUNL rows use the same 1000-iteration, seed-0 i.i.d. diagnostic on each fixed endgame; their key identity is the exact key passed to the C++ chance sampler, including root deals and, for Turn endgames, concrete river-card chance leaves. The high-revisit tabular poker cases expose longer persistent streams, whereas reduced Flop and HUNL probe lower-exposure transfer regimes. Revisit exposure is interpreted jointly with downstream weights and symmetry.

\begin{figure}[t]
\centering
\begin{minipage}[b]{0.52\textwidth}
\centering
\footnotesize
\setlength{\tabcolsep}{3.5pt}
\begin{tabular}{lccc}
\toprule
Game & Med.\ visits & \% once & Gain \\
\midrule
Kuhn            & 678 & 0\%  & \textbf{+27.6\%} \\
Leduc           & 27  & 0\%  & \textbf{+24.6\%} \\
Leduc 10 card   & 21  & 0\%  & \textbf{+34.0\%} \\
Leduc 12 card   & 15  & 0\%  & \textbf{+19.0\%} \\
Flop Hold'em (red.) & 1 & 75\% & $-2.8\%$ \\
HUNL SG1 (T) & 4 & 13\% & $+0.05\%$ \\
HUNL SG2 (T) & 4 & 23\% & $+0.00\%$ \\
HUNL SG3 (R) & 3 & 18\% & $+0.48\%$ \\
HUNL SG4 (R) & 10 & 14\% & $+0.74\%$ \\
\bottomrule
\end{tabular}
\captionof{table}{Per-chance-node revisit statistics from 1000-iteration vanilla diagnostics versus measured CCS-MCCFR gain. Statistics use observed concrete nodes, and \% once is the fraction of those nodes visited exactly once. The reduced Flop row shows the 8k endpoint; a 20-seed three-budget scan finds no detectable difference at 8k, 15k, or 50k. HUNL rows use separately instrumented fixed-endgame diagnostics.}
\label{tab:revisit}
\end{minipage}\hfill
\begin{minipage}[b]{0.45\textwidth}
\centering
\includegraphics[width=\textwidth]{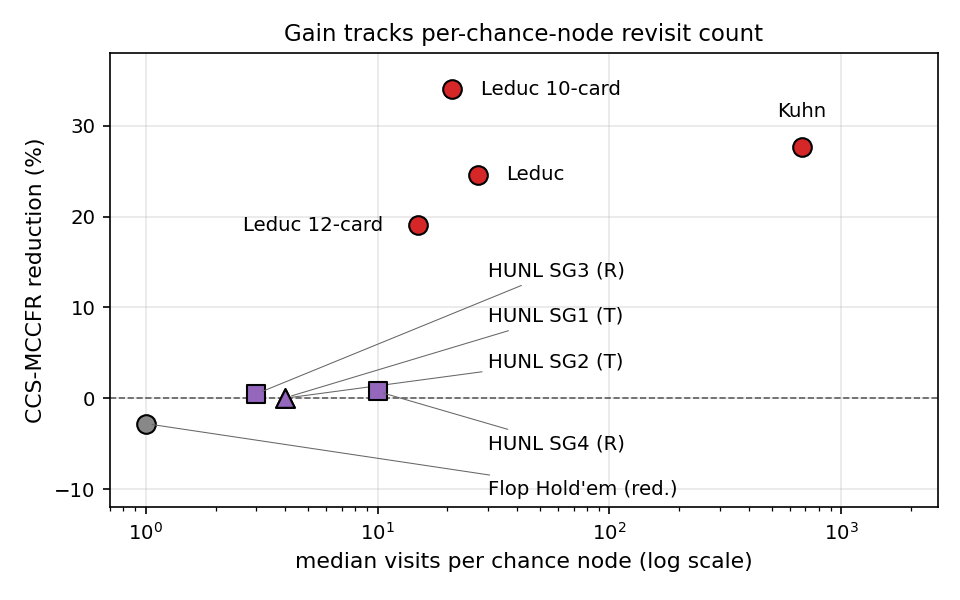}
\captionof{figure}{CCS-MCCFR gain vs.\ median visits per observed concrete chance node. The high-revisit tabular poker configurations show the largest gains; the low-exposure reduced Flop and four HUNL endgames localize regimes with no statistically detectable difference. HUNL triangles are Turn and squares are River, and Subgames~1 and~2 coincide at four median visits, so their leader lines share one marker. Per-node \% once values are listed in Table~\ref{tab:revisit}. This is a descriptive moderator visualization, not a causal or necessary-condition test.}
\label{fig:revisit}
\end{minipage}
\end{figure}

In this benchmark set, the larger gains coincide with chance nodes revisited tens to hundreds of times and with private-information coupling. The reduced Flop construction has a median of one visit and 75\% of observed nodes visited exactly once, making it a targeted low-exposure diagnostic. Across 20 paired seeds, its reductions are $-2.81\%$, $-0.38\%$, and $+2.32\%$ at 8k, 15k, and 50k node touches, respectively, and all three paired-bootstrap intervals cross zero. The budget scan therefore shows no statistically detectable deterioration or improvement rather than a stable negative effect. The construction differs structurally from full-scale Flop Hold'em and should not be read as a claim about the latter. Goofspiel is revisited often yet gains less, highlighting symmetry as a second moderator. Together these comparisons motivate a joint revisit/coupling diagnostic rather than a single-factor transfer rule.

\section{Composition Details}
\label{app:composition-details}

We test Discounted CFR and Linear CFR \citep{brown2019solving}, which time discount cumulative regret (positive regret $\times\, t^\alpha/(t^\alpha{+}1)$, negative $\times\, t^\beta/(t^\beta{+}1)$) and reach weight the average strategy ($\times\, t^\gamma$). DCFR uses $(\alpha,\beta,\gamma){=}(1.5,0,2)$ and LCFR uses $(1,1,1)$. We implement them as a solver subclass whose discount sweep touches only accumulated regrets and never the chance draw, and verify convergence against OpenSpiel's deterministic \texttt{DCFRSolver}.

In our Leduc budget sweep, the measured differences among the tested update rules narrow near $1.5\!\times\!10^6$ node touches, so we select a pre-convergence budget where the update-rule axis remains active. With i.i.d.\ chance, LCFR improves from $+12.6\%$ at $5\!\times\!10^4$ touches to $+21.5\%$ at $3\!\times\!10^5$, with all paired-bootstrap intervals clearing zero; its measured advantage narrows near $1.5\!\times\!10^6$. DCFR is harmful over this tested regime. We therefore evaluate composition at $3\!\times\!10^5$ touches.

Table~\ref{tab:orthogonal} gives the $2\times3$ Leduc grid. LCFR alone cuts exploitability by $24.2\%$ ($p{<}10^{-4}$); CCS-MCCFR cuts $19.4\%$ with vanilla updates, $24.3\%$ with LCFR, and $17.5\%$ with DCFR, with all intervals clearing zero and 98 to 100 of 100 paired seeds favoring CCS-MCCFR. LCFR plus CCS-MCCFR reaches $0.227$, a $42.59\%$ reduction from the vanilla i.i.d.\ value $0.395$.

On Kuhn ($8\!\times\!10^5$ touches, 100 seeds), neither LCFR nor DCFR accelerates vanilla at the measured endpoint, while CCS-MCCFR reduces exploitability by $+26.4\%$ with vanilla, $+24.2\%$ with LCFR, and $+18.5\%$ with DCFR, all with intervals clearing zero. Together with the Leduc grid, this shows empirical compatibility across these configurations; the differing reductions also caution against an invariant composition interpretation.

\paragraph{Composition with an External-Sampling control variate.}
\label{app:es-control-variate}

We additionally evaluate a $2\times2$ design crossing i.i.d.\ chance versus CCS-MCCFR with a control variate on the opponent action sampled by External Sampling. At an opponent information state, after sampling $a^*\sim\sigma$, the corrected node value is $\sum_a\sigma(a)b(a)+v(a^*)-b(a^*)$, where each $b(a)$ is the running mean of previously sampled continuation values for that action. This estimate is unbiased for $\sum_a\sigma(a)v(a)$ for any fixed baseline table, and disabling the baseline recovers the upstream External-Sampling recursion. The construction is motivated by the control-variate principle of VR-MCCFR \citep{schmid2019variance}, but it is restricted to the opponent-action sample in External Sampling: it is not the full recursive VR-MCCFR estimator instantiated with Outcome Sampling, which corrects every sampled value with importance weights.

Table~\ref{tab:es-control-variate} reports final endpoints from 20 paired seeds at 500k node touches for Kuhn and 1.5M for Leduc. Every entry is the reduction from the row's first condition to its second, computed by the paper's 10{,}000-resample paired bootstrap on per-seed absolute exploitability differences and normalized by the first condition's mean. With the control variate enabled, Weyl sampling reduces exploitability by 28.69\% on Kuhn and 25.64\% on Leduc, with both intervals above zero. In contrast, enabling this restricted control variate under either chance sampler yields intervals that cross zero. The experiment therefore shows that the chance-sampling effect survives this particular estimator-side modification; it neither establishes an additive interaction nor compares against the full recursive VR-MCCFR estimator instantiated with Outcome Sampling.

\begin{table}[h]
\centering
\small
\setlength{\tabcolsep}{4pt}
\caption{External-Sampling $2\times2$ composition diagnostic (20 paired seeds). Reductions and 95\% CIs use paired absolute endpoint differences normalized by the first condition's mean; wins count seeds in which the second condition has lower exploitability. Each comparison label reads first$\,\to\,$second, matching the reduction direction. ``CV'' is the restricted opponent-action control variate above, not the Outcome-Sampling instantiation of full VR-MCCFR.}
\label{tab:es-control-variate}
\begin{tabular}{llrr}
\toprule
Game & Comparison & Reduction [95\% CI] & Wins \\
\midrule
Kuhn & i.i.d.\ $\to$ Weyl, CV off & $+31.11\%\ [11.61,51.35]$ & 14/20 \\
Kuhn & i.i.d.\ $\to$ Weyl, CV on  & $+28.69\%\ [13.85,45.03]$ & 15/20 \\
Kuhn & CV off $\to$ on, i.i.d.\   & $+3.70\%\ [-16.02,23.54]$ & 9/20 \\
Kuhn & CV off $\to$ on, Weyl      & $+0.31\%\ [-24.10,24.21]$ & 11/20 \\
\midrule
Leduc & i.i.d.\ $\to$ Weyl, CV off & $+22.38\%\ [18.68,26.11]$ & 20/20 \\
Leduc & i.i.d.\ $\to$ Weyl, CV on  & $+25.64\%\ [23.24,27.98]$ & 20/20 \\
Leduc & CV off $\to$ on, i.i.d.\   & $-3.58\%\ [-8.32,1.03]$ & 9/20 \\
Leduc & CV off $\to$ on, Weyl      & $+0.78\%\ [-3.19,4.47]$ & 13/20 \\
\bottomrule
\end{tabular}
\end{table}

\section{HUNL Transfer Protocol and Results}
\label{app:hunl-details}

The HUNL diagnostic re-solves four isolated Libratus boards and ranges with a standalone C++ External-Sampling MCCFR solver. Subgames~1--2 begin on the Turn and include both root private-card deals and later river-card chance; Subgames~3--4 begin on the River and contain only root private-card chance. There is no hand abstraction, but the action space remains fixed and abstracted, so reported exploitability is within the retained abstract endgame rather than the full HUNL game. Each endgame uses $2\times10^6$ iterations and 25 paired seeds. The archive spans two solver-protocol revisions. Subgame~1 contains 6 fully-legacy pairs (seeds 0--5), 18 fully \texttt{hunl-update-v2} pairs (seeds 7--24), and one cross-protocol pair (seed 6, i.i.d.\ legacy versus Weyl v2); Subgames~2--4 use the legacy revision throughout. The v2 records explicitly expose fixed-action/no-hand-abstraction and current-visit averaging semantics. Because of this heterogeneity, and the single cross-protocol pair in Subgame~1, we interpret each endgame separately and treat the four-endgame set as a protocol-mixed transfer diagnostic rather than a homogeneous scaling estimate.

\begin{table}[h]
\centering
\small
\caption{Four HUNL endgames, direct i.i.d.\ versus CCS-MCCFR. Exploitability within the retained fixed-action abstract endgame is mean $\pm1.96$ SEM in mbb/g; reductions are paired point estimates. Every 10{,}000-resample paired-bootstrap interval crosses zero.}
\label{tab:hunl-all}
\setlength{\tabcolsep}{4pt}
\begin{tabular}{lcccc}
\toprule
Endgame & Round & i.i.d. & CCS-MCCFR & Reduction \\
\midrule
Subgame 1 & Turn  & $32427\,{\pm}\,38$ & $32411\,{\pm}\,44$ & $+0.05\%$ \\
Subgame 2 & Turn  & $39289\,{\pm}\,57$ & $39288\,{\pm}\,53$ & $+0.00\%$ \\
Subgame 3 & River & $862.4\,{\pm}\,5.5$ & $858.2\,{\pm}\,5.0$ & $+0.48\%$ \\
Subgame 4 & River & $1196.2\,{\pm}\,8.8$ & $1187.4\,{\pm}\,8.1$ & $+0.74\%$ \\
\bottomrule
\end{tabular}
\end{table}

\begin{figure}[t]
\centering
\includegraphics[width=0.9\textwidth]{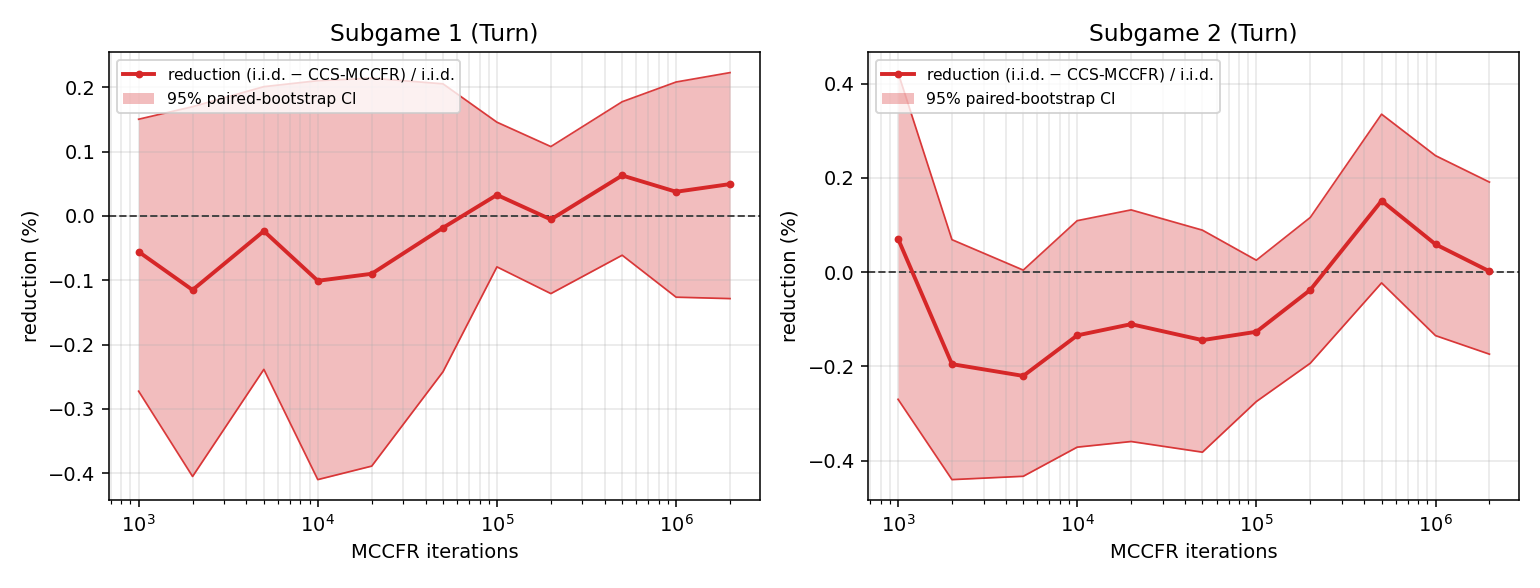}
\caption{Pointwise paired reduction on the two Turn endgames, with 95\% paired-bootstrap intervals and a zero reference line ($n{=}25$). Both bands bracket zero throughout training.}
\label{fig:hunl}
\end{figure}

\begin{figure}[t]
\centering
\includegraphics[width=0.9\textwidth]{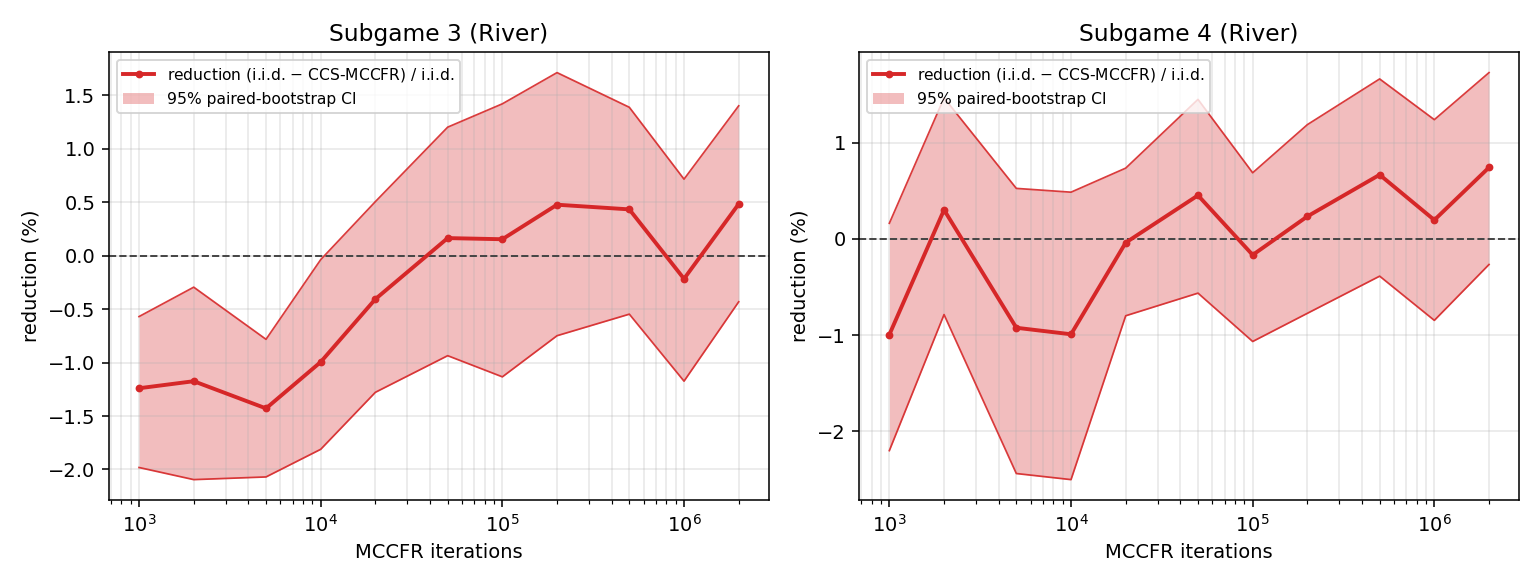}
\caption{Pointwise paired reduction on the two River endgames ($n{=}25$). Subgame~3 briefly favors i.i.d.\ at four early checkpoints before its intervals return to crossing zero; Subgame~4 brackets zero throughout, and both final intervals cross zero.}
\label{fig:hunl-river}
\end{figure}

\section{Additional Experiments and Robustness Checks}
\label{app:extra}

This appendix collects supporting results referenced in the main text: a larger symmetric-chance benchmark (Goofspiel-5), wall-clock timing, and large-budget persistence to 3M node touches. The final figures visualize the composition results detailed in Appendix~\ref{app:composition-details}.

\paragraph{Adaptive Chance Frequency Diagnostic}
\label{app:adaptive-frequency}

We measure chance frequency fidelity during the fully adaptive CCS-MCCFR run, rather than along a fixed strategy trajectory. For each concrete chance node, we record every selected outcome and the chance probability vector observed on its first visit. After 30{,}000 MCCFR iterations with seed 0, separately for Kuhn and Leduc, we compute the maximum absolute difference between empirical and true outcome probabilities at each node. Nodes with fewer than 30 visits are excluded, and the primary statistic is the visit weighted mean of these per node maximum differences. All 4 Kuhn nodes and all 157 Leduc nodes pass the visit threshold. The visit weighted mean maximum difference is $3\times10^{-5}$ for Kuhn and $6\times10^{-4}$ for Leduc; the worst single node across both games is $1.5\times10^{-2}$ in Leduc. These are descriptive diagnostics from one fully adaptive run per game, not an unbiasedness test or a confidence bound. The complementary no floor check uses 10 seeds and budgets up to 1.5M node touches for both games, together with the separate 3M touch Leduc result below.

\textbf{Goofspiel-5 (larger symmetric-chance diagnostic).} We also ran Goofspiel with 5 cards ($\approx 5\times$ the state space of Goofspiel-4), a canonical chance-heavy game used across three generations of variance-reduction work \cite{lanctot2009monte,gibson2012generalized,farina2020stochastic}. At 1M node touches (6 paired seeds), CCS-MCCFR yields $0.6465 \pm 0.012$ versus vanilla $0.6412 \pm 0.007$, a $-0.8\%$ point estimate (two-sided paired $t$-test $p = 0.37$). Given the small sample, this is a descriptive boundary diagnostic. It is consistent with, but does not establish, the hypothesis that symmetric public chance leaves less residual chance variation for temporal balancing to affect in regret differences.

\textbf{Wall clock check.} CCS-MCCFR computes $u_n = (\phi_c + ng) \mod 1$ and an inverse CDF lookup per chance-node visit. On Leduc at 1.5M node touches (5 paired seeds, single core), mean $\pm$ SD solver time is $25.47 \pm 0.19$\,s for vanilla and $24.66 \pm 1.84$\,s for CCS-MCCFR, with mean ratio $0.968$ (two-sided paired $t$-test $p = 0.42$). The check detects no timing difference and is intended as a lightweight implementation diagnostic rather than an equivalence test.

\textbf{Large-budget persistence.} At 3M Leduc node touches (10 paired seeds), vanilla reaches $0.09026$ exploitability and CCS-MCCFR $0.06813$, a 24.5\% reduction (two-sided paired $t$-test $p = 1.9\times10^{-6}$). This closely matches the 24.59\% main Leduc result and shows that the measured advantage persists at twice the principal budget.

Figures~\ref{fig:dcfr-accel} and~\ref{fig:stack-grid} visualize the budget sweep and composition grid discussed in Appendix~\ref{app:composition-details}.

\begin{figure}[h]
\centering
\includegraphics[width=0.7\textwidth]{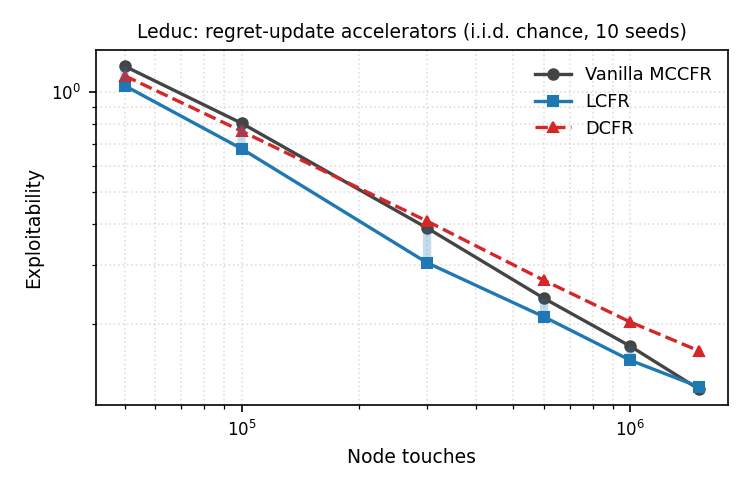}
\caption{MCCFR regret update accelerators vs.\ node touch budget on Leduc (i.i.d.\ chance, 10 seeds). LCFR is significantly faster than vanilla through $\sim\!10^6$ touches (shaded gap); DCFR's $\beta{=}0$ discount is harmful under Monte Carlo regret noise. The update rule axis is active only before convergence.}
\label{fig:dcfr-accel}
\end{figure}

\begin{figure}[h]
\centering
\includegraphics[width=0.6\textwidth]{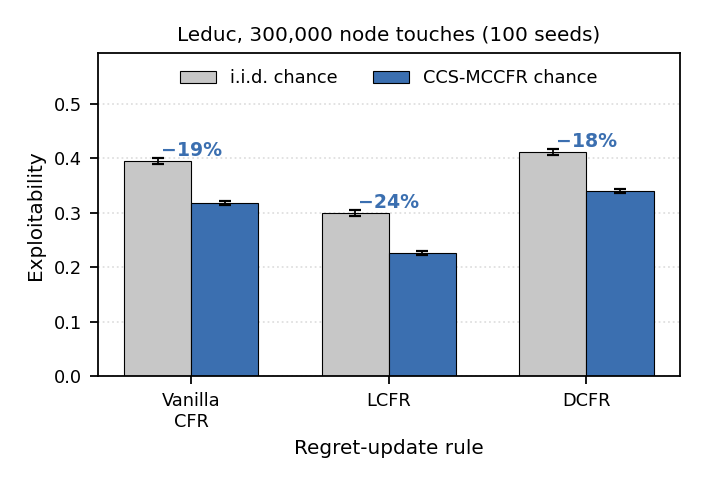}
\caption{Leduc at $3\times10^5$ node touches (100 seeds). Within each tested update rule, CCS-MCCFR (blue) lowers exploitability relative to i.i.d.\ chance (grey); across rules, LCFR has the lowest i.i.d.\ exploitability. LCFR plus CCS-MCCFR is the lowest measured cell. The differing relative reductions across rules indicate empirical compatibility rather than an additive or invariant composition law. Error bars are paired bootstrap 95\% CIs.}
\label{fig:stack-grid}
\end{figure}

\end{document}